\documentclass[10pt,conference,compsocconf,letterpaper]{IEEEtran}
\usepackage{comment}
\usepackage{graphicx}
\usepackage{subfigure}
\usepackage{algorithm}
\usepackage{amsmath}
\usepackage{algorithmic}

\def\baselinestretch{0.95}

\newtheorem{lemma}{\bf Lemma}
\newtheorem{theorem}[lemma]{\bf Theorem}

\ifCLASSINFOpdf
\else
\fi
\begin{document}
%
\title{Dynamic Deferral of Workload for Capacity Provisioning in Data Centers}

\author{\IEEEauthorblockN{Muhammad Abdullah Adnan\IEEEauthorrefmark{1}, Ryo Sugihara\IEEEauthorrefmark{2}, Yan Ma\IEEEauthorrefmark{3} and Rajesh K. Gupta\IEEEauthorrefmark{1}}
\IEEEauthorblockA{\IEEEauthorrefmark{1}University of California San Diego, \IEEEauthorrefmark{2}Amazon.com, \IEEEauthorrefmark{3}Shandong University}
}

\maketitle

\begin{abstract}
Recent increase in energy prices has led researchers to find better ways for capacity provisioning in data centers to reduce energy wastage due to the variation in workload. This paper explores the opportunity for cost saving utilizing the flexibility from the Service Level Agreements (SLAs) and proposes a novel approach for capacity provisioning under bounded latency requirements of the workload. We investigate how many servers to be kept active and how much workload to be delayed for energy saving while meeting every deadline. We present an offline LP formulation for capacity provisioning by dynamic deferral and give two online algorithms to determine the capacity of the data center and the assignment of workload to servers dynamically. We prove the feasibility of the online algorithms and show that their worst case performance are bounded by a constant factor with respect to the offline formulation. We validate our algorithms on a MapReduce workload by provisioning capacity on a Hadoop cluster and show that the algorithms actually perform much better in practice compared to the naive `follow the workload' provisioning, resulting in 20-40\% cost-savings.
\end{abstract}


%

\section{Introduction}
With the advent of cloud computing, data centers are emerging all over the world and their energy consumption becomes significant; as estimated 61 million MWh per year, costing about 4.5 billion dollars \cite{23}. Naturally, energy efficiency in data centers has been pursued in various ways including the use of renewable energy \cite{22,25} and improved cooling efficiency \cite{26,30,17}, etc. Among these, improved scheduling algorithm is a promising approach for its broad applicability regardless of hardware configurations. Among the attempts to improve scheduling \cite{17, 9}, recent effort has focussed on optimization of schedule under performance constraints imposed by Service Level Agreements (SLAs). Typically, a SLA specification provides a measure of flexibility in scheduling that can be exploited to improve performance and efficiency \cite{n9,n10}. To be specific, latency is an important performance metric for any web-based service and is of great interest for service providers who run their services on data centers.
The goal of this paper is to utilize the flexibility from the SLAs for different types of workload to reduce energy consumption. The idea of utilizing SLA information to improve performance and efficiency is not entirely new. Recent work explores utilization of application deadline information for improving the performance of the applications (e.g. see \cite{n10, n6}). But the opportunities for energy efficiency remain unexplored, certainly in a manner that seeks to establish bounds on the energy cost from the proposed solutions.


In this paper, we are interested in minimizing the energy consumption of a data center under guarantees on latency/ deadline. We use the deadline information to defer some tasks so that we can reduce the total cost for energy consumption for executing the workload and switching the state of the servers. We determine the portion of the released workload to be executed at the current time and the portions to be deferred to be executed at later time slots without violating deadlines. Our approach is similar to `valley filling' that is widely used in data centers to utilize server capacity during the periods of low loads \cite{9}. But the load that is used for valley filling is mostly background/maintenance tasks (e.g. web indexing, data backup) which is different from actual workload. In fact current valley filling approaches ignore the workload characteristics for capacity provisioning. In this paper, we determine how much work to defer for valley filling in order to reduce the current and future energy consumption while provably ensuring satisfaction of SLA requirements. Later we generalize our approach for more general workloads where different jobs have different deadlines.

This paper makes three contributions. First, we present an LP formulation for capacity provisioning with dynamic deferral of workload. The formulation not only determines capacity but also determines the assignment of workload for each time slot. As a result the utilization of each server can be determined easily and resources can be allocated accordingly. Therefore this method well adapts to other scheduling policies that take into account dynamic resource allocation, priority aware scheduling, etc.

Second, we design two optimization based online algorithms depending on the nature of the deadline. For uniform deadline, our algorithm named {\it Valley Filling with Workload (VFW($\delta$))}, looks ahead $\delta$ slots to optimize the total energy consumption. The algorithm uses the valley filling approach to defer some workload to execute in the periods of low loads. For nonuniform deadline, we design a {\it Generalized Capacity Provisioning (GCP)} algorithm that reduces the switching (on/off) of servers by balancing the workloads in adjacent time slots and thus reduces energy consumption.  We prove the feasibility of the solutions and show that the performance of the online algorithms are bounded by a constant factor with respect to the offline formulation.

Third, we validate our algorithms using MapReduce traces (representative workload for data centers) and evaluate cost savings achieved via dynamic deferral. We run simulations to deal with a wide range of settings and show significant savings in each of them. Over a period of 24 hours, we find more than 40\% total cost saving for GCP and around 20\% total cost saving for VFW($\delta$) even for small deadline requirements. We compare the two online algorithms with different parameter settings and find that GCP gives more cost savings than VFW($\delta$). In order to show that our algorithms work on real systems, we perform experiments on a 35 node Hadoop cluster and find energy savings of $\sim$6.02\% for VFW($\delta$) and $\sim$12\% for GCP over a period of 4 hours. The experimental results show that the peak energy consumption for the operation of a data center can be reduced by provisioning capacity and scheduling workload using our algorithms.




The rest of the paper is organized as follows. Section II presents the model that we use to formulate the optimization and gives the offline formulation. In Section III, we present the VFW($\delta$) algorithm for determining capacity and workload assignment dynamically when the deadline is uniform. In Section IV, we illustrate the GCP algorithm with nonuniform deadline. Section V discusses the extension of our algorithms for long jobs. In Section VI and VII, we illustrate the simulation and experimental results respectively. In Section VIII, we describe the state of the art research related to capacity provisioning and Section IX concludes the paper.

\section{Model Formulation}
In this section, we describe the model we use for capacity provisioning via dynamic deferral. We note that the assumptions used in this model are minimal and this formulation captures many properties of current data center capacity and workload characteristics.

\subsection{Workload Traces}
To build a realistic model, we need real workload from data centers but the data center providers are reluctant to publish their production traces due to privacy issues and competitive concerns. To overcome the scarcity of publicly available traces, efforts have been made to extract summary statistics from production traces and workload generators based on those statistics have been proposed \cite{n4,n5}. For the purposes of this paper, we use such a workload generator and use the MapReduce traces released by Chen et al \cite{n4}. MapReduce framework is widely used in Data centers and acts as representative workload where each of the jobs consists of 3 steps of computation: map, shuffle and reduce \cite{n3}. Figure~\ref{fig:original_workload}(a) illustrates the statistical  MapReduce traces over 24 hours generated from real Facebook traces.


\begin{figure}[!t]
\centerline{\subfigure[Original Workload]{\includegraphics[width
=1.7in]{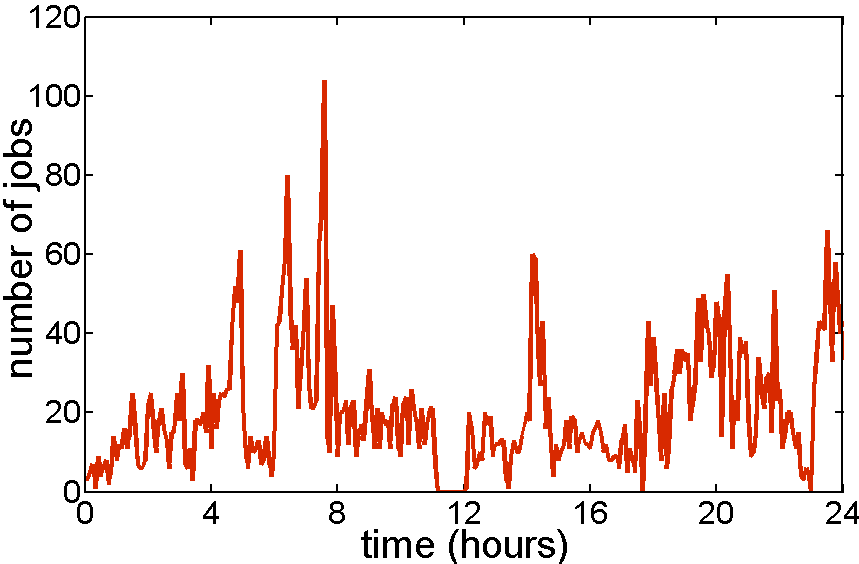}
}
\hfil
\subfigure[Batch and Interactive Job]{\includegraphics[width=1.7in]{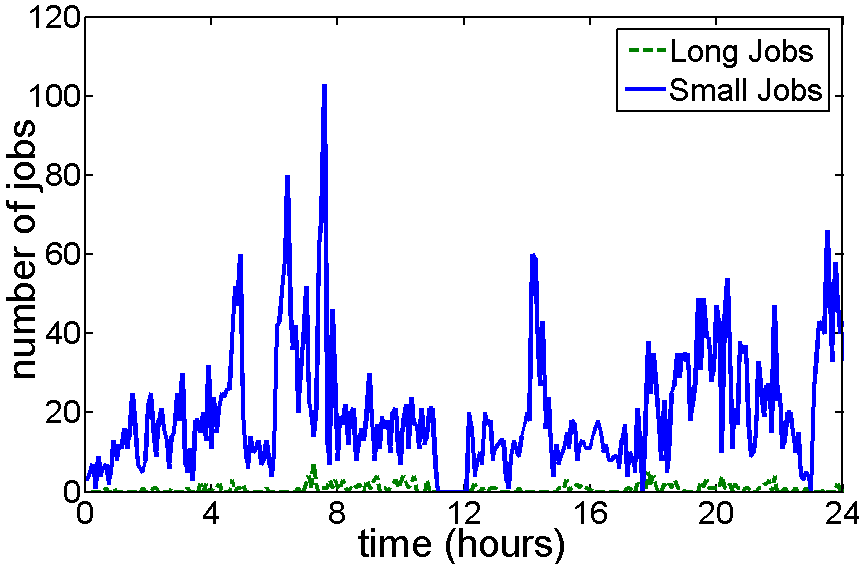}
}}
\caption{Illustration of (a) original workload and (b) distinction between batch and small interactive jobs.
}
\label{fig:original_workload}
\end{figure}

Typically the workload traces consist of a mix of batch and interactive jobs. Chen et al. carried out an interactive analysis to classify the workload and showed that the workload is dominated ($\sim$98\%) by small and interactive jobs showing significant and unpredictable variation with time. Table~\ref{table:trace} illustrates the classification on the MapReduce traces by $k$-means clustering based on the sizes of map, shuffle and reduce stages (in bytes) with $k=10$ and Figure~\ref{fig:original_workload}(b) shows the distinction in time variation between the long batch jobs and small interactive jobs. To adapt with the large variation in the small and interactive workload, valley filling methods have been proposed using the low priority batch jobs to fill in the periods of low workload \cite{Liu}. However, Chen et al. have shown that the portion of low priority long jobs ($\sim$2\%) are insufficient to reduce the variation (to smooth) in the workload curve \cite{n5}. In this paper, we propose valley filling with workload (mix of long and interactive jobs) and devise algorithms for capacity provisioning by scheduling jobs under bounded latency requirements.


\begin{table}[!t]
\caption{Cluster Sizes and Medians by $k$-means Clustering on the MapReduce Trace} 
\centering  
\begin{tabular}{r r r r r} 
\hline\hline                        
\# Jobs & \% Jobs & Input & Shuffle & Output \\ [0.5ex] 
\hline                  
5691 & 96.56 & 15 KB & 0 & 685 KB \\ 
116 & 1.97 & 44 GB & 15 GB  & 84 MB  \\
27 & 0.46 & 56 GB & 145 GB & 16 GB \\
23 & 0.39 & 123 GB & 0 & 52 MB \\
19 & 0.32 & 339 KB & 0 & 48 GB\\ [1ex]      
8 & 0.14 & 203 GB & 404 GB & 3 GB \\
5 & 0.08 & 529 GB & 0 & 53 KB \\
3 & 0.05 & 46 KB & 0 & 199 GB\\
1 & 0.02 & 7 TB & 48 GB & 101 GB \\
1 & 0.02 & 913 GB& 8 TB & 61 KB \\
\hline 
\end{tabular}
\label{table:trace} 
\end{table}












\subsection{Workload Model}
We consider a workload model where the total workload varies over time. The time interval we are interested in is $t\in \{0,1,\ldots,T\}$ where $T$ can be arbitrarily large. In practice, $T$ can be a year and the length of a time slot $\tau$ could be as small as 2 minutes (the minimum time required to change power state of a server). In our model, each job has a deadline $D$ (in terms of number of slots) associated with it, where $D$ is a nonnegative integer. In other words, a job released at time $t$, needs to be executed within time slot $t + D$. We first model for small interactive jobs having length less than $\tau$. Later in Section V, we extend our model for mix of jobs with arbitrary lengths. Based on the nature of deadlines, we have two cases: (i) uniform deadline, when deadline is uniform for all the jobs; (ii) non-uniform deadlines, when different jobs have different deadlines. In Section II and III, we formulate model and algorithm for the case of uniform deadline and the non-uniform deadline case is considered in Section IV.
Let $L_t$ be the amount of workload released at time slot $t$. Since the deadline $D$ is uniform for all the jobs, the total amount of work $L_t$ must be executed by the end of time slot $t + D$. Since $L_t$ varies over time, we often refer to it as a {\it workload curve}.

We consider data center as a collection of homogeneous servers. The total number of servers $M$ is fixed and given but each server can be turned on/off to execute the workload. We normalize $L_t$ by the processing capability of each server i.e. $L_t$ denotes the number of servers required to execute the workload at time $t$. We assume for all $t$, $L_t\le M$. Let $x_{i,d,t}$ be the portion of the released workload $L_t$ that is assigned to be executed at server $i$ at time slot $t+d$ where $d$ represents the deferral with $0\le d\le D$. Let $m_t$ be the number of active servers during time slot $t$. Then $\sum_{i=1}^{m_t} \sum_{d=0}^{D} x_{i,d,t} = L_t \text{ and } 0\le x_{i,d,t}\le 1$.

Let $x_{i,t}$ be the total workload assigned at time $t$ to server $i$ and $x_t$ be the total assignment at time $t$. Then we can think of $x_{i,t}$ as the utilization of the $i$th server at time $t$ i.e. $0\le x_{i,t}\le 1$. Thus $\sum_{d=0}^{D} x_{i,d,t-d} = x_{i,t} \text{ and } \sum_{i=1}^{m_t} x_{i,t} = x_t$. From the data center perspective, we focus on two important decisions during each time slot $t$: (i) determining $m_t$, the number of active servers, and (ii) determining $x_{i,d,t}$, assignment of workload to the servers.

\subsection{Cost Model}
The goal of this paper is to minimize the cost (price) of energy consumption in data centers. The energy cost function consists of two parts: operating cost and switching cost. {\it Operating cost} is the cost for executing the workload which in our model is proportional to the assigned workload. We use the common model for energy cost for typical servers which is an affine function:
$$C(x)=e_0+e_1x$$
where $e_0$ and $e_1$ are constants (e.g. see \cite{20}) and $x$ is the assigned workload (utilization) of a server at a time slot. Although we use this general model for cost function, other models considering nonlinear parameters such as temperature, frequency can easily be adopted in the model which will make it a nonlinear optimization problem. Our algorithms can be applied for such nonlinear models by using techniques for solving nonlinear optimizations as each optimization is considered as a single independent step in the algorithms.


{\it Switching cost $\beta$} is the cost incurred for changing state (on/off) of a server. We consider the cost of both turning on and turning off a server. Switching cost at time $t$ is defined as follows: $$S_t = \beta |m_t-m_{t-1}|$$ where $\beta$ is a constant (e.g. see \cite{9,21}).

\subsection{Optimization Problem}
Given the models above, the goal of a data center is to choose the number of active servers (capacity) $m_t$ and the dispatching rule $x_{i,d,t}$ to minimize the total cost during $[1,T]$, which is captured by the following optimization:
\begin{IEEEeqnarray}{lll}
\label{equn:opt1}
 \text{min}_{x_t,m_t}\quad & \sum_{t=1}^T \sum_{i=1}^{m_t}  C(x_{i,t}) + \beta \sum_{t=1}^T |m_t-m_{t-1}|\\
 \text{subject to}\quad &  \sum_{i=1}^{m_t} \sum_{d=0}^{D} x_{i,d,t} = L_t  \qquad\qquad\qquad \forall t\nonumber\\
 &  \sum_{i=1}^{m_t} \sum_{d=0}^{D} x_{i,d,t-d} \le m_t  \qquad\qquad\quad \forall t\nonumber\\
 &  \sum_{d=0}^{D} x_{i,d,t-d} \le 1  \qquad\qquad\qquad\forall i, \forall t\nonumber\\
 &  0\le m_t \le M \qquad\qquad\qquad\qquad \forall t\nonumber\\
 &  x_{i,d,t}\ge 0  \qquad\qquad\qquad\qquad\forall i, \forall d, \forall t.\nonumber
\end{IEEEeqnarray}
Since the servers are identical, we can simplify the problem by dropping the index $i$ for $x$. More specifically, for any feasible solution $x_{i,d,t}$, we can make another solution by $x_{i,d,t}=\sum_{i=1}^{m_t} x_{i,d,t}/m_t$ (i.e., replacing every $x_{i,d,t}$ by the average of $x_{i,d,t}$ for all $i$) without changing the value of the objective function while satisfying all the constraints after this conversion. Then we have the following optimization equivalent to~(\ref{equn:opt1}):
\begin{IEEEeqnarray}{cll}
\label{equn:opt2}
 \text{min}_{x_t,m_t}\quad & \sum_{t=1}^T  m_tC(x_t/m_t) + \beta \sum_{t=1}^T |m_t-m_{t-1}|&\\
 \text{subject to}\quad &  \sum_{d=0}^{D} x_{d,t} = L_t   &\forall t \nonumber\\
  &  \sum_{d=0}^{D} x_{d,t-d} \le m_t   &\forall t\nonumber\\
  &  0\le m_t \le M  &\forall t\nonumber\\
  &  x_{d,t}\ge 0    &\forall d, \forall t.\nonumber
\end{IEEEeqnarray}
where $x_{d,t}$ represents the portion of the workload $L_t$ to be executed at a server at time $t+d$. We further simplify the problem by showing that any optimal assignment for (\ref{equn:opt2}) can be converted to an equivalent assignment that uses earliest deadline first (EDF) policy (see Figure~\ref{fig:lemma1}). More formally, we have the following lemma:

\begin{lemma}
\label{lemma:edf}  Let $x^*_{t_r}$ and $x^*_{t_s}$ be the optimal assignments of workload obtained from the solution of optimization~(\ref{equn:opt2}) at times $t_r$ and $t_s$ respectively where $t_s>t_r$ and $t_s-t_r = \theta < D$. If $\exists \delta$ with $\sum_{d=0}^{\delta-1} x^*_{d, t_r-d}\ne 0$ and $\sum_{d=\theta+\delta+1}^{D} x^*_{d, t_s-d}\ne 0$ for any $0<\delta<(D-\theta)$ then we can obtain another assignments $x^{e}_{t_r}= x^*_{t_r}$ and $x^{e}_{t_s}= x^*_{t_s}$ where $\sum_{d=0}^{\delta-1} x^e_{d, t_r-d}= 0$ and $\sum_{d=\theta+\delta+1}^{D} x^e_{d, t_s-d}= 0$.
\end{lemma}

\begin{proof}
We prove it by constructing $x^{e}_{t_r}$ and $x^{e}_{t_s}$ from $x^*_{t_r}$ and $x^*_{t_s}$. We change the assignments $x^*_{d,t_r}$, $0\le d \le (D-\theta)$ and $x^*_{d,t_s}$, $\theta \le d \le D$ to obtain $x^{e}_{t_r}$ and $x^{e}_{t_s}$ as illustrated in Figure~\ref{fig:lemma1}.
\begin{figure}[!ht]
\begin{center}
\includegraphics[width=3.3in]{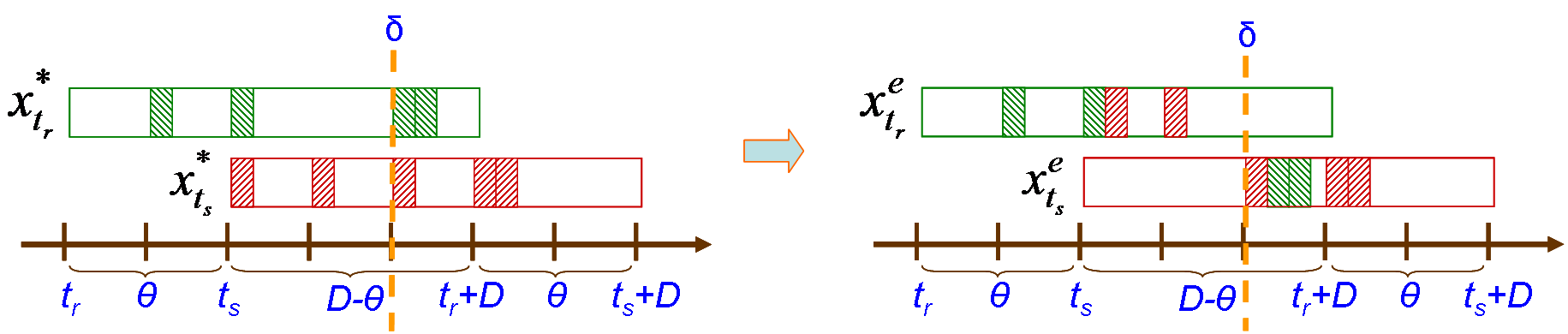}
\caption{Assignments can be determined from their release times and EDF policy.}
\label{fig:lemma1}
\end{center}
\end{figure}
We now determine $\delta$. Note that all the workloads released between (including) time slots $t_s-D$ to $t_r$ can be executed at time $t_r$ without violating deadline since $t_r-D<t_s-D<t_r-\delta <t_r$. Also all the workloads released between (including) time slots $t_s-D$ to $t_r$ can be executed at time $t_s$ without violating deadline since $t_s-D<t_r-\delta <t_r<t_s$. Hence the new assignment of workloads cannot violate any deadline. We determine $\delta$ at a point where $\sum_{d=\delta+1}^{D-\theta} x^e_{d, t_r-d} = \sum_{d=\delta+1}^{D-\theta} x^*_{d, t_r-d} + \sum_{d=\theta+ \delta+1}^{D} x^*_{d, t_s-d}$ and $\sum_{d=0}^{\delta-1} x^e_{d, t_r-d}= 0$ and $x^e_{\delta, t_r-\delta} = \sum_{d=0}^{D-\theta} x^*_{d,t_r} - \sum_{d=\delta+1}^{D-\theta} x^e_{d, t_r-d}$ such that $x^{e}_{t_r}= x^*_{t_r}$. Similarly for $x^{e}_{t_s}$, we have the new assignment as: $\sum_{d=\theta}^{\theta+\delta-1} x^e_{d, t_s-d}$ $=$
$\sum_{d=0}^{\delta-1} x^*_{d, t_r-d}$ $ + \sum_{d=\theta}^{\theta+ \delta-1} x^*_{d, t_s-d}$ and $\sum_{d=\theta+ \delta+1}^{D} x^e_{d, t_s-d}$ $= 0$ and $x^e_{\theta+\delta, t_s-\theta-\delta} = \sum_{d=\theta}^{D} x^*_{d,t_s} - \sum_{d=\theta}^{\theta+\delta-1} x^e_{d, t_s-d}$ such that $x^{e}_{t_s}= x^*_{t_s}$.
\end{proof}

According to lemma~\ref{lemma:edf}, we do not need both $t$ and $d$ as indices of $x$. We can use the release time $t$ to determine the deadline $t+D$ and differentiate between the jobs using their deadlines. Thus, we drop the index $d$ of $x$. At time $t$, unassigned workload from $L_{t-D}$ to $L_t$ is executed according to EDF policy while minimizing the objective function. To formulate the constraint that no assignment violates any deadline we define delayed workload $l_t$ with maximum deadline $D$.
\begin{equation*}
l_t =
\begin{cases}
0 & \text{if } t \le D,\\
L_{t-D} & \text{otherwise.}
\end{cases}
\end{equation*}
We call the delayed curve $l_t$ for the workload as {\it deadline curve}. Thus we have two fundamental constraints on the assignment of workload for all $t$:

\begin{enumerate}
\item[(C1)] Deadline Constraint: $\sum_{j=1}^{t} l_j \le \sum_{j=1}^{t} x_j$
\item[(C2)] Release Constraint: $\sum_{j=1}^{t} x_j \le \sum_{j=1}^{t} L_j$
\end{enumerate}

Condition (C1) says that all the workloads assigned up to time $t$ cannot violate deadline and Condition (C2) says that the assigned workload up to time $t$ cannot be greater than the total released workload up to time $t$. Using these constraints we reformulate the optimization~(\ref{equn:opt2}) as follows:

\begin{IEEEeqnarray}{cLl}
\label{equn:opt3}
 \text{min}_{x_t,m_t}\quad & \sum_{t=1}^T  m_tC(x_t/m_t) + \beta \sum_{t=1}^T |m_t-m_{t-1}|&\\
 \text{subject to}\quad & \sum_{j=1}^{t} l_j \le \sum_{j=1}^{t} x_j \le \sum_{j=1}^{t} L_j  & \forall t\nonumber\\
 &  \sum_{j=1}^{T} x_j = \sum_{j=1}^{T} L_j   & \nonumber\\
 &  0\le x_t \le m_t \le M   & \forall t\nonumber
\end{IEEEeqnarray}

Since the operating cost function $C(.)$ is an affine function, the objective function is linear as well as the constraints.  Hence it is clear that the optimization (\ref{equn:opt3}) is a linear program. Note that capacity $m_t$ in this formulation is not constrained to be an integer. This is acceptable because data centers consists of thousands of active servers and we can round the resulting solution with minimal increase in cost. Figure~\ref{fig:solutionrandom}(a) illustrates the offline optimal solutions for $x_t$ and $m_t$ for a dynamic workload generated randomly. The performance of the optimal offline solution on two realistic workload are provided in Section VI.

\begin{figure}[!t]
\centerline{\subfigure[Offline optimal]{\includegraphics[width
=1.7in]{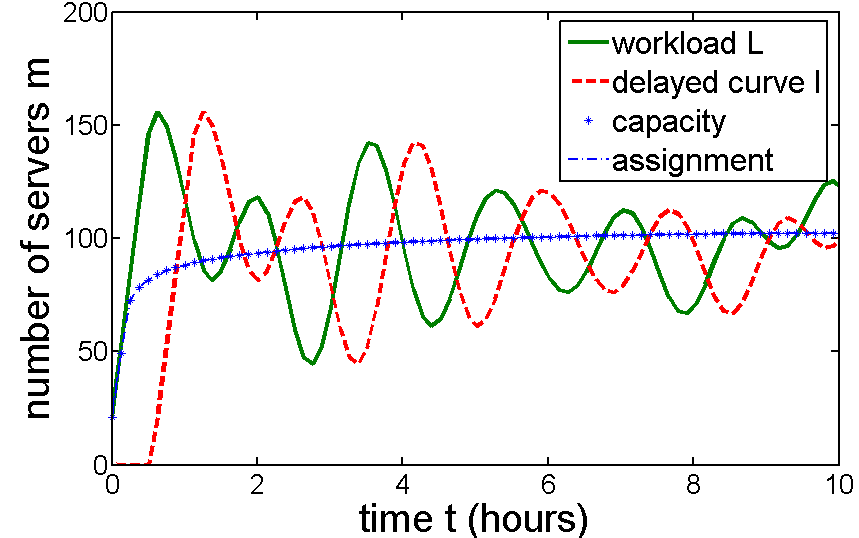}
}
\hfil
\subfigure[VFW($\delta$)]{\includegraphics[width=1.7in]{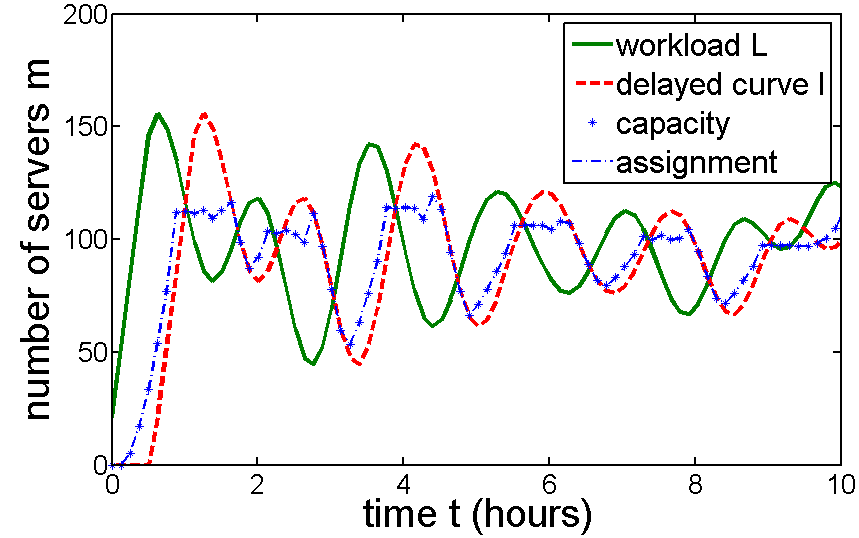}
}}
\caption{Illustration of (a) offline optimal solution and (b) VFW($\delta$) for arbitrary workload generated randomly; time slot length = 2 min, $D=15$, $\delta=10$.}
\label{fig:solutionrandom}
\end{figure}

\section{Valley Filling with Workload}
In this section, we consider the online case, where at any time $t$, we do not have information about the future workload $L_{t'}$ for $t'>t$. At each time $t$, we determine the $x_t$ and $m_t$  by applying optimization over the already released unassigned workload which has deadline in future $D$ slots. Note that the workload released at or before $t$, can not be delayed to be assigned after time slot $t+D$. Hence we do not optimize over more than $D+1$ slots.  We simplify the online optimization by solving only for $m_t$ and determine $x_t$ by making $x_t = m_t$ at time $t$. This makes the online algorithm not to waste any execution capacity that cannot be used later for executing workload. But the cost due to switching in the online algorithm may be higher than the offline algorithm. Thus our goal is to design strategies to reduce the switching cost. In the online algorithm, we reduce the switching cost by optimizing the total cost for the interval $[t,t+D]$.

When the deadline is uniform, we can reduce the switching cost even more by looking beyond $D$ slots. We do that by accumulating some workload from periods of high loads and execute that amount of workload later in valleys without violating constraints (C1) and (C2). Note that by accumulation we do not violate deadline as at each slot, we execute a portion of the accumulated workload by swapping with the newly released workload by EDF policy. To determine the amount of accumulation and execution we use `$\delta$-delayed workload'. Thus the online algorithm namely Valley Filling with Workload (VFW($\delta$)) looks ahead $\delta$ slots to determine the amount of execution. Let $l^{\delta}_t$ be the {\it $\delta$-delayed curve} with delay of $\delta$ slots for $0<\delta <D$.
\begin{equation*}
l^\delta_t =
\begin{cases}
0 & \text{if } t \le \delta,\\
L_{t-\delta} & \text{otherwise.}
\end{cases}
\end{equation*}
Then we can call the deadline curve as {\it $D$-delayed curve} and represent it by $l^D_t$. We determine the amount of accumulation and execution by controlling the set of feasible choices for $m_t$ in the optimization. For this purpose, we use the $\delta$-delayed curve to restrict the amount of accumulation. By having a lower bound on $m_t$ for the valley (low workload) and an upper bound it for the high workload, we control the execution in the valley and accumulation in the other parts of the curve. Thus in the online algorithm, we have two types of optimizations: Local Optimization and Valley Optimization. Local Optimization is used to smooth the `wrinkles' (we define {\it wrinkles} as the small variation in the workload in adjacent slots e.g. see Figure~\ref{fig:valleyopt}) within $D$ consecutive slots and accumulate some workload. On the other hand, Valley Optimization fills the valleys with the accumulated workload.

\begin{figure}[!t]
\begin{center}
\includegraphics[width=3.3in]{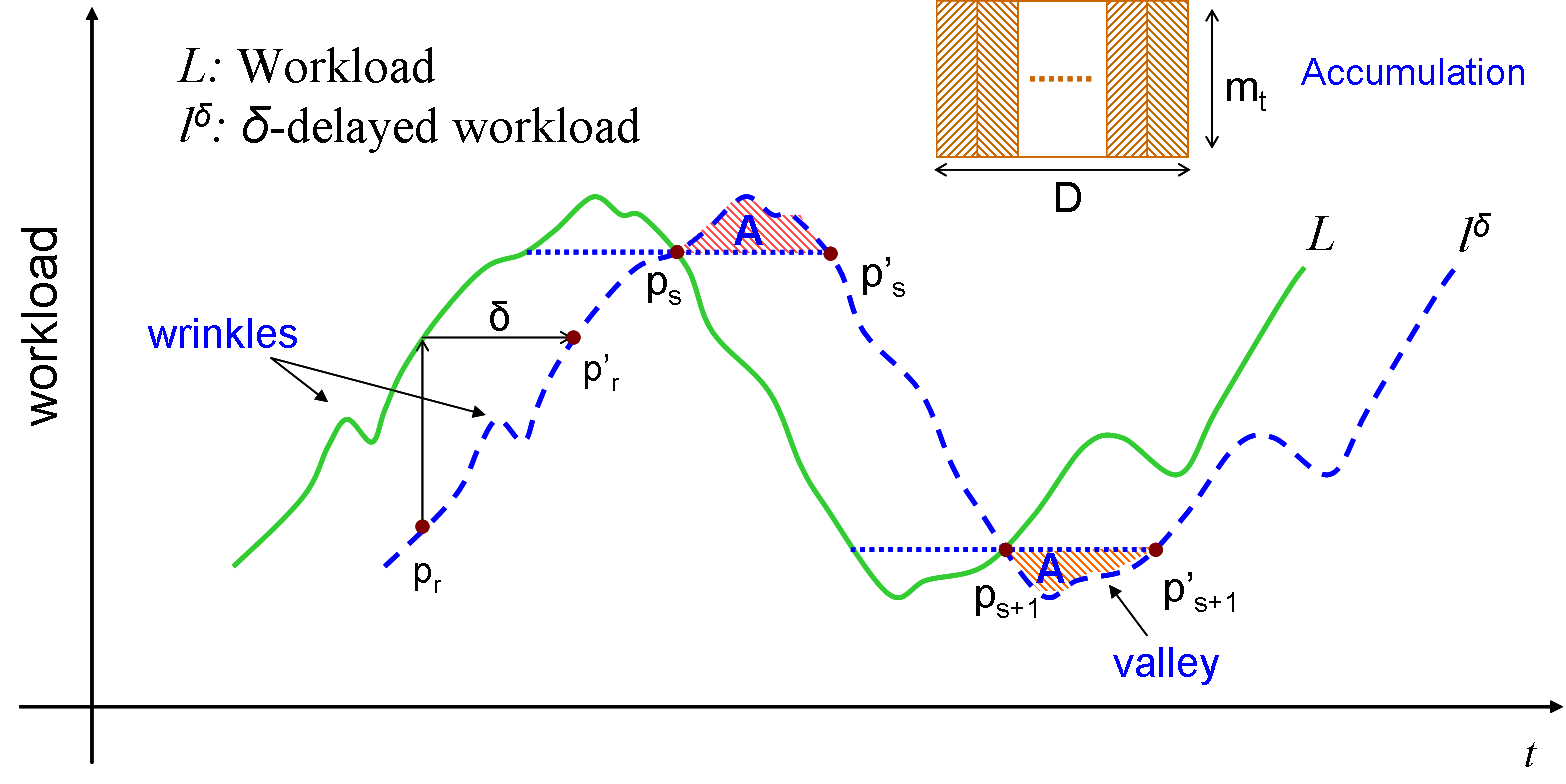}
\caption{The curves $L_t$ and $l^\delta_t$ and their intersection points. The peak from the $l^\delta_t$ curve is cut and used to fill the valley of the same curve. The amount of workload that is accumulated/delayed is bounded by $m_tD$.}
\label{fig:valleyopt}
\end{center}
\end{figure}

\subsection{Local Optimization}
The {\it local optimization} applies optimization over future $D$ slots and finds the optimum capacity for current slot by executing no more than $\delta$-delayed workload. Let $t$ be the current time slot. At this slot we apply a slightly modified version of offline optimization (\ref{equn:opt3}) in the interval $[t, t+D]$. We apply the following optimization LOPT($l_t$, $l^\delta_t$, $m_{t-1}$, $M$) to determine $m_t$ in order to smooth the wrinkles by optimizing over $D$ consecutive slots. We restrict the amount of execution to be no more than the $\delta$-delayed workload while satisfying the deadline constraint (C1).
\begin{IEEEeqnarray}{cLl}
\label{lopt}
 \text{min}_{m_t}\quad & (e_0 + e_1)\sum_{j=t}^{t+D}  m_j + \beta \sum_{j=t}^{t+D} |m_j-m_{j-1}|\\
 \text{subject to}\quad &\sum_{j=1}^{t} l^D_j\le \sum_{j=1}^{t} m_j     \nonumber\\
 &  \sum_{j=1}^{t+D} m_j = \sum_{j=1}^{t} l^\delta_j  \nonumber\\
 &  0\le m_k \le M  \qquad \qquad \qquad t\le k \le t+D\nonumber
\end{IEEEeqnarray}
After solving the local optimization, we get the value of $m_{t}$ for the current time slot and assign $x_{t}=m_{t}$.  For the next time slot $t+1$ we solve the local optimization again to find the values for $x_{t+1}$ and $m_{t+1}$. Note that the deadline constraint (C1) and the release constraint (C2) are satisfied at time $t$, since from the formulation $\sum_{j=1}^{t} l^D_j\le \sum_{j=1}^{t} m_j \le \sum_{j=1}^{t} l^\delta_j \le \sum_{j=1}^{t} L_j$.

\subsection{Valley Optimization}
In {\it valley optimization}, the accumulated workload from the local optimization is executed in `global valleys'. Before giving the formulation for the valley optimization we need to detect a valley.

Let $p_1, p_2, \ldots, p_n$ be the sequence of intersection points of $L_t$ and $l^\delta_t$ curves (see Figure~\ref{fig:valleyopt}) in nondecreasing order of their x-coordinates ($t$ values).  Let $p'_1, p'_2, \ldots, p'_n$ be the sequence of points on $l^\delta_t$ with delay $\delta$ added with each intersection point $p_1, p_2, \ldots, p_n$ on $l^\delta_t$ such that $t'_s =t_s + \delta$ for all $1\le s\le n$. We discard  all the intersection points (if any) between $p_s$ and $p'_s$ from the sequence such that $t_{s+1} \ge t'_s$. Note that at each intersection point $p_s$, the curve from $p_s$ to $p'_s$ is known. To determine whether the curve $l^\delta_t$ between $p_s$ and $p'_s$ is a valley, we calculate the area $$A = \sum_{t=t_s}^{t'_s} (l^\delta_t-l^\delta_{t_s})$$
If $A$ is negative, then we regard the curve between $p_s$ and $p'_s$ as a {\it global valley} though it may contain several small peaks and valleys. If the curve between $p_s$ and $p'_s$ is a global valley, we fill the valley with some (possibly all) of the accumulated workload by executing more than the $\delta$-delayed workload while satisfying the release constraint (C2). For each $t$, we apply the following optimization  VOPT($l_t$, $L_t$, $m_{t-1}$, $M$) in the interval $[t,t+D]$ to find the value of $m_t$ where $t_s\le t\le t'_s$.
\begin{IEEEeqnarray}{cLl}
\label{equn:gopt}
 \text{min}_{m_t}\quad & (e_0 + e_1) \sum_{j=t}^{t+D}  m_j + \beta \sum_{j=t}^{t+D} |m_j-m_{j-1}|\\
 \text{subject to}\quad & \sum_{j=1}^{t} l^D_j\le \sum_{j=1}^{t} m_j \quad \quad\quad\quad\quad\quad \quad   \nonumber\\
 &  \sum_{j=1}^{t+D} m_j = \sum_{j=1}^{t} L_j  \nonumber\\
 &  0\le m_k \le M  \quad \quad\quad\quad\quad\quad \quad t\le k \le t+D\nonumber
\end{IEEEeqnarray}
Note that the deadline constraint (C1) and the release constraint (C2) are satisfied at time $t$, since $\sum_{j=1}^{t} l^D_j\le \sum_{j=1}^{t} m_j$ $\le \sum_{j=1}^{t} L_j$. We apply the valley optimization (\ref{equn:gopt}) for  each $t_s\le t\le t'_s$ and local optimization (\ref{lopt}) for each time slot $t$ where $t \in \{[1,T-D-1] - [t_s, t'_s]\}$ for all $t_s$. For each $t\in[T-D,T]$ we apply the valley optimization (\ref{equn:gopt}) for global valley in the interval $[t,T]$ in order to execute all the accumulated workload. Algorithm~\ref{VFW} summarizes the procedures for VFW($\delta$). For each new time slot $t$, Algorithm~\ref{VFW} detects a valley by checking whether the curves $l^\delta_t$ and $L_t$ intersects. If $t$ is inside a valley, Algorithm~\ref{VFW} applies valley optimization (VOPT); local optimization (LOPT), otherwise. Figure~\ref{fig:solutionrandom}(b) illustrates the nature of solutions from VFW($\delta$) for $x_t$ and $m_t$. Note that $\delta$ is a parameter for the online algorithm VFW($\delta$).

\begin{algorithm}
\caption{VFW($\delta$)}
\label{VFW}
{\small{
\begin{algorithmic}[1]
\STATE $valley \leftarrow 0$; $m_0 \leftarrow 0$
\STATE $l^D[1:D] \leftarrow 0$; $l^\delta[1:\delta] \leftarrow 0$
\FOR {{\bf each} new time slot $t$}
\STATE $l^D[t+D] \leftarrow L[t]$
\STATE $l^\delta[t+\delta] \leftarrow L[t]$
\IF {$valley = 0$ and $l^\delta$ intersects $L$ }
\STATE Calculate Area $A = \sum_{t=t_s}^{t'_s} (l^\delta_t-l^\delta_{t_s})$
\IF {$A<0$}
\STATE $valley \leftarrow 1$
\ENDIF
\ELSIF {$valley > 0$ and $valley \le \delta$}
\STATE $valley \leftarrow valley +1$
\ELSE
\STATE $valley \leftarrow 0$
\ENDIF
\IF{$valley = 0$}
\STATE $m[t:t+D]$ $\leftarrow$ LOPT($l[1:t]$,$l^\delta[1:t]$,$m_{t-1}$,$M$)
\ELSE
\STATE $m[t:t+D]$ $\leftarrow$ VOPT($l[1:t]$,$L[1:t]$,$m_{t-1}$,$M$)
\ENDIF
\STATE $x_t \leftarrow m_t$
\ENDFOR
\end{algorithmic}
}}
\end{algorithm}

\subsection{Analysis of the Algorithm}
We first prove the feasibility of the solutions from the VFW($\delta$) algorithm and then analyze the competitive ratio of this algorithm with respect to the offline formulation~(\ref{equn:opt3}). First, we have the following theorem about the feasibility.

\begin{theorem}
\label{theorem:online3}
The VFW($\delta$) algorithm gives feasible solution for any $0<\delta<D$.
\end{theorem}

\begin{proof}
We prove this theorem inductively by showing that the choice of any feasible $m_t$ from an optimization applied in the interval $[t,t+D]$ do not result in infeasibility in the optimization applied in the next time slot $[t+1,t+D+1]$. Initially, the optimization in VFW($\delta$) is applied for the interval $[1,D+1]$ with $\sum_{j=1}^{k} l^D_j = 0$ for $1\le k \le D$. Hence the optimization applied in the intervals $[1,D+1]$ gives feasible  $m_1$ because $\sum_{j=1}^{k} l^D_j \le \sum_{j=1}^{k} l^\delta_j \le \sum_{j=1}^{k} L_j$ for $1\le k \le D$.

Now suppose the VFW($\delta$) gives feasible $m_t$ in an interval $[t,t+D]$. We have to prove that there exists feasible choice for $m_t$ for the optimization applied at $[t+1,t+D+1]$. The deadline constraint (C1) and the release constraint (C2) are satisfied for $m_t$. Hence, $\sum_{j=1}^{t} l^D_j \le \sum_{j=1}^{t} l^\delta_j \le \sum_{j=1}^{t} L_j$. Since $0<\delta<D$, $\sum_{j=1}^{t} l^D_j \le \sum_{j=1}^{t+1} l^D_j \le \sum_{j=1}^{t} l^\delta_j \le \sum_{j=1}^{t+1} l^\delta_j \le \sum_{j=1}^{t} L_j\le \sum_{j=1}^{t+1} L_j$. Thus for any feasible choice of $m_t$, we can always obtain feasible solution for $m_{t+1}$ such that the above inequality holds.
\end{proof}

We now analyze the competitive ratio of the online algorithm with respect to the offline formulation~(\ref{equn:opt3}). We denote the operating cost of the solution vectors $X = (x_1,x_2,\ldots,x_T)$ and $M = (m_1,m_2,\ldots,m_T)$ by $cost_o(X,M)=\sum_{t=1}^T m_t \cdot$ $C(x_t/m_t)$, switching cost by $cost_s(X,M) = \beta \sum_{t=1}^T |m_t-m_{t-1}|$ and total cost by $cost(X,M) = cost_o(X,M)$ $+$ $cost_s(X,M)$. We have the following lemma.

\begin{lemma}
\label{lemma:online24}
$cost_s(X,M) \le 2\beta\sum_{t=1}^{T} m_t$
\end{lemma}

\begin{proof}
Switching cost at time $t$ is $S_t = \beta|m_t-m_{t-1}|\le \beta(m_t + m_{t-1})$, since $m_t\ge 0$. Then $cost_s(X,M)\le \beta \cdot$ $\sum_{t=1}^T (m_t + m_{t-1}) \le  2\beta \sum_{t=1}^T m_t$ where $m_0 = 0$.
\end{proof}

Let $X^*$ and $M^*$ be the offline solution vectors from optimization (\ref{equn:opt3}). The following theorem proves that the competitive ratio of the VFW($\delta$) algorithm is bounded by a constant with respect to the offline formulation~(\ref{equn:opt3}).

\begin{theorem}
\label{theorem:online1}
$cost(X,M) \le \frac{e_0+e_1+2\beta}{e_0 + e_1}cost(X^*,M^*)$.
\end{theorem}

\begin{proof}
Since the offline optimization assigns all the workload in the $[1,T]$ interval, $\sum_{t=1}^T x^*_t = \sum_{t=1}^T L_t \le \sum_{t=1}^T m^*_t$, where we used $x^*_t\le m^*_t$ for all $t$. Hence $cost(X^*,M^*) \ge cost_o(X^*,M^*) = \sum_{t=1}^T m^*_tC(x^*_t/m^*_t) = \sum_{t=1}^T (e_0 m^*_t + e_1 x^*_t) \ge  \sum_{t=1}^T (e_0 + e_1) L_t$.

In the online algorithm, we set $x_t = m_t$ and $\sum_{j=1}^t m_j$ $\le \sum_{j=1}^t L_j$ for all $t\in[1,T]$. Hence by lemma~\ref{lemma:online24}, we have $cost(X,M) = cost_o(X,M)+ cost_s(X,M) \le \sum_{t=1}^T (e_0 + e_1) m_t + 2\beta\sum_{t=1}^{T} m_t  \le (e_0+ e_1) \sum_{t=1}^T  L_t + 2\beta\sum_{t=1}^T L_t = (e_0 + e_1 +2\beta)\sum_{t=1}^T L_t$.
\end{proof}


\section{Generalized Capacity Provisioning}
We now consider the general case where the deadline requirements are not same for all the jobs in a workload. Let $\nu$ be the maximum possible deadline.
We decompose the workload according to their associated deadline. Suppose $L_{d,t}\ge 0$ be the portion of the workload released at time $t$ and has deadline $d$ for $0\le d\le \nu$. We have $\sum_{d=0}^{\nu} L_{d,t} = L_t$.

The workload to be executed at any time slot $t$ can come from different previous slots $t-d$ where $0\le d\le \nu$ as illustrated in Figure~\ref{fig:fine-grained}(a). Hence we redefine the deadline curve $l_t$ and represent it by $l'_t$. Assuming $L_{d,t} = 0$ if $t\le 0$, we define $l'_t = \sum_{d=0}^\nu L_{d,(t-d)}$. Then the offline formulation remains the same as formulation (\ref{equn:opt3}) with the deadline curve $l_t$ replaced by $l'_t$.
\begin{IEEEeqnarray}{cLl}
\label{equn:finegrained}
 \text{min}_{x_t,m_t}\quad & \sum_{t=1}^T  m_tC(x_t/m_t) + \beta \sum_{t=1}^T |m_t-m_{t-1}| &\\
 \text{subj. to}\quad & \sum_{j=1}^{t} l'_j \le \sum_{j=1}^{t} x_j \le \sum_{j=1}^{t} L_j  & \forall t  \nonumber\\
 &  \sum_{j=1}^{T} x_j = \sum_{j=1}^{T} L_j  &  \nonumber\\
 &  0\le x_t \le m_t \le M  &\forall t  \nonumber
\end{IEEEeqnarray}

We now consider the online case. Delaying the workload up to their maximum deadline may increase the switching cost since it may increase the variation in the workload compared to the original workload (see Figure~\ref{fig:fine-grained}(b)). Hence at each time we need to determine the optimum assignment and capacity that reduces the switching cost from the original workload while satisfying each individual deadline. We can apply the VFW($\delta$) algorithm from the previous section with $D=D_{min}$ where $D_{min}$ is the minimum deadline for the workload. If $D_{min}$ is small, VFW($\delta$) does not work well because $\delta<D_{min}$ becomes too small to detect a valley. Hence we use a novel approach for distributing the workload $L_t$ over the $D_t$ slots such that the change in the capacity between adjacent time slots is minimal (see Figure~\ref{fig:fine-grained}(c)). We call this algorithm as Generalized Capacity Provisioning (GCP) algorithm.

\begin{figure}[!t]
\begin{center}
\includegraphics[width=3.4in]{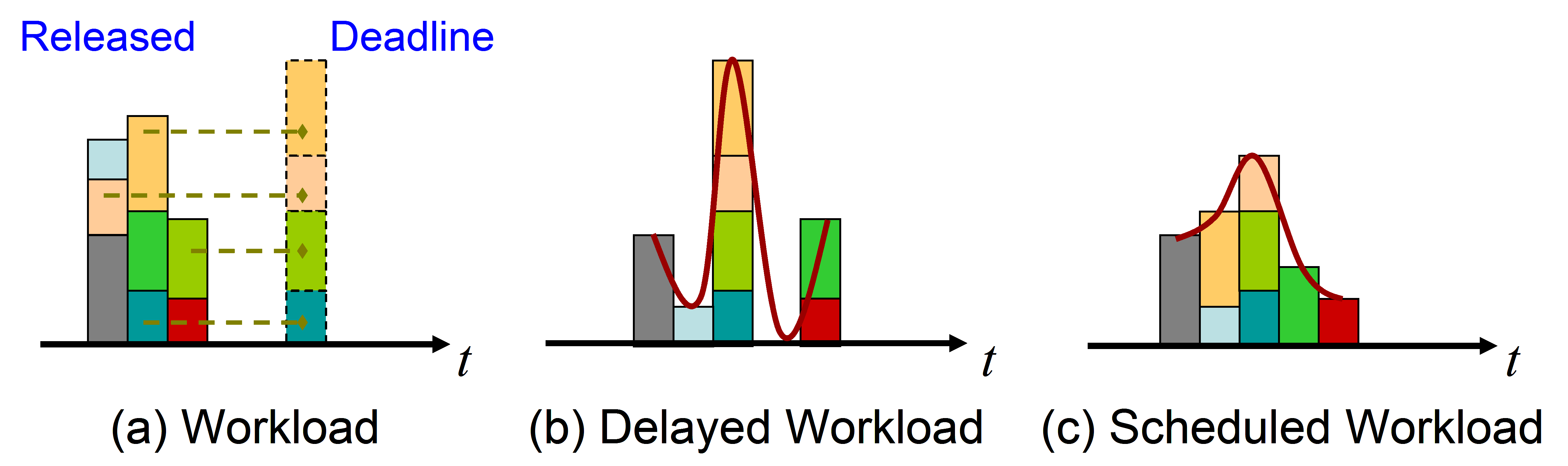}
\caption{Illustration of workload with different deadline requirements. (a) workload released at different times have different deadlines, (b) the delayed workload $l'_t$, may increase the switching cost due to large variation, (c) distribution of workload in adjacent slots by GCP to reduce the variation in workload.}
\label{fig:fine-grained}
\end{center}
\end{figure}

In the GCP algorithm, we apply optimization to determine $m_t$ at each time slot $t$ and make $x_t = m_t$. The optimization is applied over the interval $[t,t+\nu]$ since at time slot $t$ we can have workload that has deadline up to $t+\nu$ slots. Hence at each time $t$, the released workload is a vector of $\nu+1$ dimension. Let, ${\bf L_t} = (L_{0,t},L_{1,t},\ldots, L_{\nu, t})$ where $L_{d,t} = 0$ if there is no workload with deadline $d$ at time $t$. Let ${\bf y_t}$ be the vector of unassigned workload released up to time $t$. The vector ${\bf y_t}$ is updated from ${\bf y_{t-1}}$ at each time slot by subtracting the capacity $m_{t-1}$ and then adding ${\bf L_t}$.  Note that $m_{t-1}$ is subtracted from the vector  ${\bf y_{t-1}}$ in order to use unused capacity  to execute already released workload at time $t-1$ by following EDF policy (see lines 4-17 in Algorithm~\ref{GCP}).  Let ${\bf y'_{t-1}} = (y'_{0,t-1},y'_{1,t-1},y'_{2,t-1},\ldots, y'_{\nu, t-1})$ be the vector after subtracting $m_{t-1}$ with $y'_{0,t-1} =0$ and $y'_{j,t-1}\ge 0$ for $1\le j\le \nu$. Then ${\bf y_t} = {\bf L_t} + (y'_{1,t-1},y'_{2,t-1},\ldots, y'_{\nu, t-1},0)$ where ${\bf y_t} = (0, 0,\ldots, 0)$ if $t<=0$. Then the optimization GCP-OPT(${\bf y_t}$, $m_{t-1}$, $M$) applied at each $t$  over the interval $[t,t+\nu]$ is as follows:
\begin{IEEEeqnarray*}{cll}
\label{equn:fine-grained}
 \text{min}_{m_t} & (e_0 + e_1)\sum_{j=t}^{t+\nu} m_{j} + \beta \sum_{j=t}^{t+\nu} |m_{j}-m_{j-1}|& \quad \quad \IEEEyessubnumber\\
 \text{subject to} & \quad  \sum_{j=0}^{\nu} m_{t+j} = \sum_{j=0}^{\nu} y_{j,t}  & \IEEEyessubnumber \\
& \quad \sum_{k=0}^{j} m_{t+k}  \ge \sum_{k=0}^{j} y_{k,t}  \qquad 0\le j\le \nu-1 & \IEEEyessubnumber\\
&  \quad 0\le m_{t+j} \le M  \qquad \qquad 0\le j \le \nu & \IEEEyessubnumber
\end{IEEEeqnarray*}

Note that the optimization (\ref{equn:fine-grained}) solves for $\nu+1$ values. We only use $m_t$ as the capacity and assignment of workload at time $t$. Algorithm~\ref{GCP} summarizes the procedures for GCP. The GCP algorithm gives feasible solutions because it works with the unassigned workload and constraint (7c) ensures deadline constraint (C1) and constraint (7b) ensures the release constraint (C2). The competitive ratio for the GCP algorithm is same as the competitive ratio for VFW($\delta$) because in GCP, $m_t=x_t$ and release constraint (C2) holds at every $t$ making $\sum_{t=1}^T m_t = \sum_{t=1}^T x_t \le \sum_{t=1}^T L_t$.

\begin{algorithm}[!t]
\caption{GCP}
\label{GCP}
{\small{
\begin{algorithmic}[1]
\STATE $y[0:\nu] \leftarrow 0$
\STATE $m_0 \leftarrow 0$
\FOR {{\bf each} new time slot $t$}
\STATE $uc \leftarrow m_{t-1}$          \COMMENT{$uc$ represents the unused capacity}
\FOR {$i=0$ to $\nu$}
\IF {$uc \le 0$ }
\STATE $y'[i] \leftarrow y[i]$
\ELSE
\STATE $uc \leftarrow uc - y[i]$
\IF {$uc \le 0$ }
\STATE $y'[i] \leftarrow -uc$
\ELSE
\STATE $y'[i] \leftarrow 0$
\ENDIF
\ENDIF
\ENDFOR
\STATE $y[0:\nu] = \{y'[1:\nu],0\} + L_t[0:\nu]$
\STATE $m[t:t+\nu]$ $\leftarrow$ GCP-OPT($y[0:\nu]$, $m_{t-1}$, $M$)
\STATE $x_t \leftarrow m_t$
\ENDFOR
\end{algorithmic}
}}
\end{algorithm}

\section{Extension for Long Jobs}
In this section, we extend our model for `Long Jobs'. By {\it long jobs} we mean the jobs that have length greater than the time slot length $\tau$ and each of which has a given deadline requirement greater than its length. To extend the model for long jobs we estimate the length of the long jobs, decompose the long jobs into small pieces ($\le\tau$) and assign deadline to each of them.


\subsection{Estimation of Execution Time}
If all the jobs are short interactive jobs with execution time less than the time-slot length, then we do not need to estimate the completion time of the jobs. But for a mix of short and long jobs we need estimation.

Since our targeted workload is MapReduce, we present a method for estimating execution time of a MapReduce job.  The MapReduce performance model is illustrated in Figure~\ref{fig:mapreduce}. A MapReduce job $J$ is defined in \cite{Estimation_MapReduce} as a 7 tuple $(S,S',S'',X,Y,f(x),g(x))$, where $S$ is the size of map input data; $S'$ is the size of intermediate shuffle data; $S''$ is the size of output reduce data; $X$ is the number of mappers that $J$ is divided into; $Y$ is the number of reducers assigned for $J$ to output; $f(x)$ is the running time of a mapper with input size $x$; $g(x)$ is the running time of a reducer with input size $x$.  We compute the number of map and reduce tasks by dividing the input size $S$ and output size $S''$ by the HDFS (Hadoop Distributed File System) block size respectively. Let $V_i$, $V_o$ and $V_n$ be the data read rate, data output rate and network transfer rate respectively. Then the execution times of map, shuffle and reduce tasks can be estimated by the following equations \cite{Estimation_MapReduce}.
\begin{IEEEeqnarray*}{ll}
\mathcal{T}_{m} = \frac{S}{XV_i} + f(\frac{S}{X}) + \frac{S'}{XV_o}&\\
\mathcal{T}_{s} = \frac{S'}{XYV_n}&\\
\mathcal{T}_{r} = g(\frac{S'}{Y}) + \frac{S''}{YV_o}&
\end{IEEEeqnarray*}

\begin{figure}[!t]
\begin{center}
\includegraphics[width=3.3in]{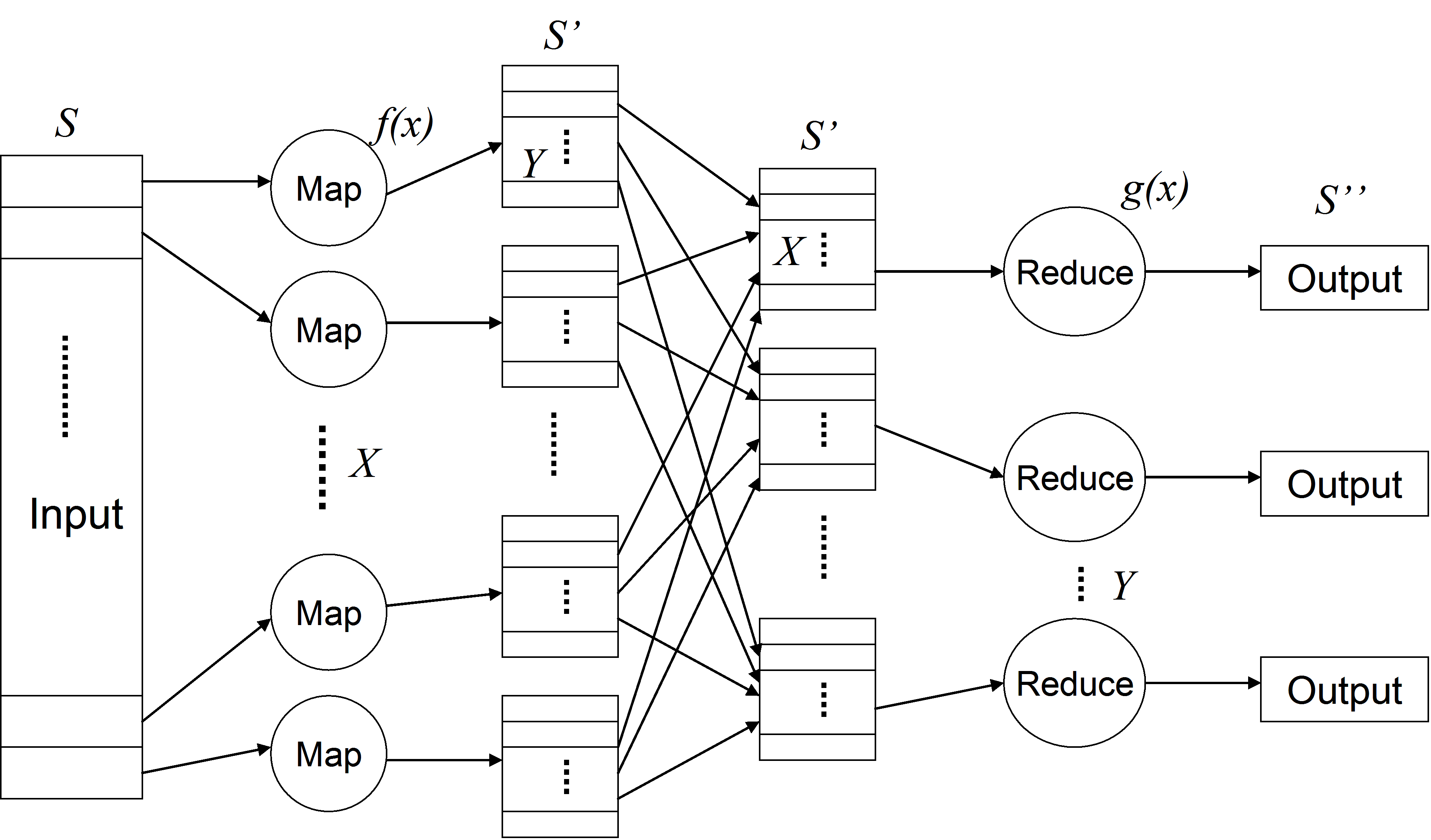}
\caption{MapReduce Performance Model.}
\label{fig:mapreduce}
\end{center}
\end{figure}

The map and reduce tasks run in parallel and may need several rounds to finish if the total number of available task slots is less than the number of mappers or reducers. If at most $M$ mappers and $R$ reducers are completed at each round, then the number of rounds of map and reduce tasks are $\lambda_m = \lceil X/M\rceil$ and $\lambda_r = \lceil Y/R\rceil$ respectively. Based on whether the reducers have to wait for data transfer to complete, we have two cases. If $t_m\ge Mt_s$, reducers have to wait for data transfers to complete and when $t_m < Mt_s$, reducers do not need to wait for completion of mappers except during the first round. Then the estimated execution time for job $J$ is,
\begin{equation*}
\mathcal{T}_J =
\begin{cases}
\mathcal{T}_{m} + \lambda_r(X\mathcal{T}_{s} + \mathcal{T}_{r}), & \mathcal{T}_{m} < M \mathcal{T}_{s} \\
\lambda_m \mathcal{T}_{m} + M \mathcal{T}_{s} + \lambda_r ( X \mathcal{T}_{s} + \mathcal{T}_{r}), & \mathcal{T}_{m} \ge M\mathcal{T}_{s}
\end{cases}
\end{equation*}



\subsection{Decomposition and Deadline Assignment}
The long jobs can be preemptive and non-preemptive. Based on these two types we have two ways to decompose and assign deadlines to them.

\subsubsection*{Preemptive Jobs}
Let $J$ be a preemptive job with execution time $\mathcal{T}_J$ ($>\tau$), release time $t_J$ and deadline $D$ (in terms of number of slots). Then the length of the job $J$ is $\ell = \lceil \mathcal{T}_J/\tau \rceil$ time slots. Since the job is preemptive, we safely decompose it into small pieces $J_1, J_2, \ldots J_\ell$ where $\mathcal{T}_{J_i}\le \tau$ for $1\le i\le \ell$. We assign the deadline for each of the small pieces to $\lfloor D/\ell \rfloor-1$ (see Figure~\ref{fig:longjob}(a)). We set the release time of the first piece $t_{J_1} = t_J$ and the release times of the other pieces to $t_{J_i} = t_{J_{i-1}}+\lfloor D/\ell \rfloor$ for $1<i\le \ell$. Since a piece $J_{i}$ is released after its preceding piece $J_{i-1}$, this technique satisfies any precedence constraint the particular job may have.

\begin{figure}[!ht]
\centerline{\subfigure[Preemptive Jobs]{\includegraphics[width
=1.6in]{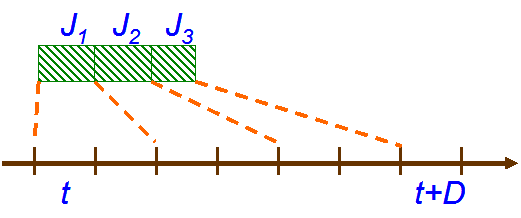}
}
\hfil
\subfigure[Non-preemptive Jobs]{\includegraphics[width=1.6in]{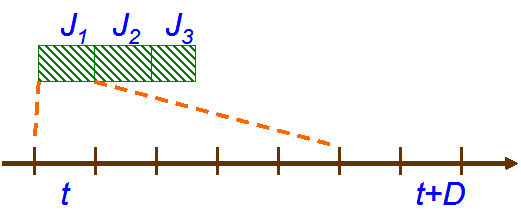}
}}
\caption{Release times and deadlines for pieces of Long Jobs.}
\label{fig:longjob}
\end{figure}






\subsubsection*{Non-preemptive Jobs}
To make scheduling decision for a job $J$, we consider it to be consisting of $J_1, J_2, \ldots J_\ell$ small pieces where $\mathcal{T}_{J_i}= \tau$ for $1\le i< \ell$ and $\mathcal{T}_{J_\ell}\le \tau$. We schedule the first chunk $J_1$ with deadline $D-\ell$ and suppose that our algorithm schedules it at time $t'_{J_1}$ (see Figure~\ref{fig:longjob}(b)). Then we assign the release time of the other pieces $J_i$ for $1<i\le \ell$ to be $t_{J_i} = t'_{J_{i-1}}+1$ and assign deadline $D= 0$ for each of them. Since the deadline is zero for the later pieces and they are released as soon as the previous piece finishes, the algorithm ensures enough capacity for the execution of the job without preemption.




\section{Simulation}
In this section, we evaluate the cost incurred by the VFW($\delta$) and GCP algorithm relative to optimal solution in the context of workload generated from realistic data. First, we motivate our evaluation by a detailed analysis of simulation results. Then in Section VII, we validate the simulation results by performing experiments on a Hadoop cluster.

\subsection{Simulation Setup}
We use realistic parameters in the simulation setup and provide conservative estimates of cost savings resulting from our proposed VFW($\delta$) and GCP algorithms.

\subsubsection*{Cost benchmark}
Currently data centers typically do not use dynamic capacity provisioning based on the variation of the workload \cite{9}. A naive approach for capacity provisioning is to follow the workload curve and determine the capacity and assignment of workload accordingly. Clearly it is not a good approach because for capacity provisioning it does not take into account the cost incurred due to switching. Yet this is a very conservative estimate as it does not waste any execution capacity and meets all the deadline. We compare the total cost from the VFW($\delta$) and GCP algorithms with the `follow the workload' ($x=m=L$) strategy and evaluate the cost reduction.

\subsubsection*{Cost function parameters}
The total cost is characterized by $e_0$ and $e_1$ for the operating cost and $\beta$ for the switching cost. In the operating cost, $e_0$ represents the proportion of the fixed cost and $e_1$ represents the load dependent energy consumption.  The energy consumption of the current servers is dominated by the fixed cost \cite{19}. Therefore we choose $e_0 = 1$ and $e_1 = 0$. The switching cost parameter $\beta$ represents the wear-and-tear due to changing power states in the servers. We choose $\beta =12$ for slot length of 5 minutes such that it works as an estimate of the time a server should be powered down (typically one hour \cite{9,21}) to outweigh the switching cost with respect to the operating cost.

\subsubsection*{Workload description}
We use two publicly available MapReduce traces as examples of dynamic workload. The MapReduce traces were released by Chen et al. \cite{n4} which are produced from real Facebook traces for one day (24 hours) from a cluster of 600 machines. We count the number of different types of job submissions over a time slot length of 5 minutes to find the released workload curve (Figure~\ref{fig:workload}). We then use the length estimation to find the actual workload curve showing active jobs over time and use that as a dynamic workload (Figure~\ref{fig:workload}) for simulation. The two samples we use represent strong diurnal properties and have variation from typical workload (Workload A) to bursty workload (Workload B). The released jobs will be delayed to reduce the variations in the active workload curve.

\begin{figure}[!t]
\centerline{\subfigure[Workload A]{\includegraphics[width
=1.7in]{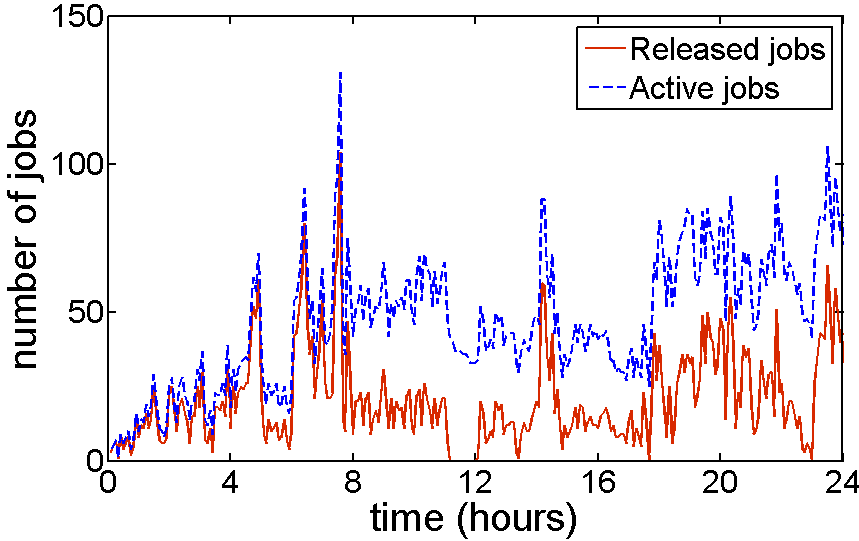}
}
\hfil
\subfigure[Workload B]{\includegraphics[width=1.7in]{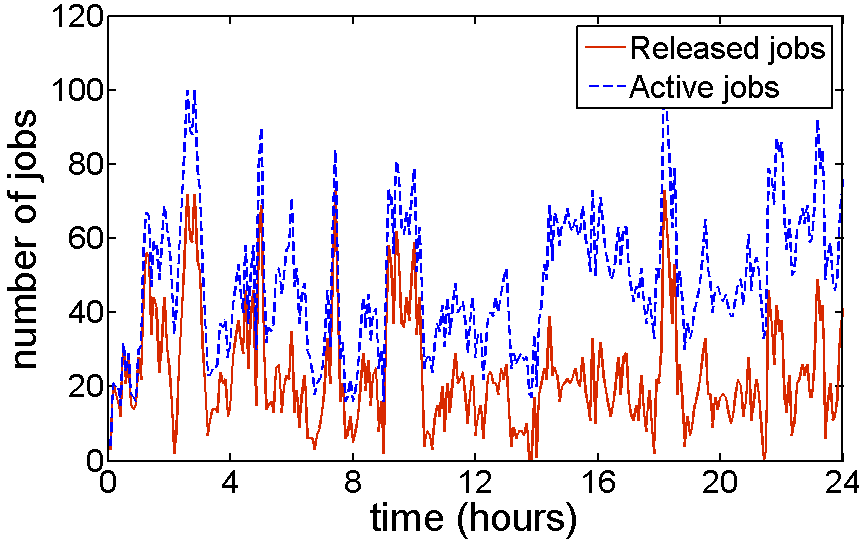}
}}
\caption{Illustration of the MapReduce traces as dynamic workload used in the experiments. The active jobs are shown with job length estimation.}
\label{fig:workload}
\end{figure}

\subsubsection*{Length estimation parameters}
Since we do not know any parameters from the Facebook cluster, we assume that the cluster has enough task slots so that all the map and reduce tasks are completed in the first round. Hence $\lambda_m= 1$ and $\lambda_r =1$. HDFS block size is typically 128MB. So $X = \lceil \frac{S}{128\textrm{MB}} \rceil$ and $Y=\lceil\frac{S''}{128\textrm{MB}}\rceil$. We assume $f(x)$ and $g(x)$ to be linear, $f(x)=\alpha_1 x$ and $g(x)= \alpha_2 x$ where $\alpha_1 = 0.8$s/MB and $\alpha_2 = 0.9$s/MB \cite{Estimation_MapReduce}. We use typical values for data rates, $V_i = 100\textrm{MB/s}$, $V_o = 100\textrm{MB/s}$ and $V_n = 10\textrm{MB/s}$ \cite{n9,Estimation_MapReduce}. For our experiments, we consider the jobs to be non-preemptive. Hence we do not need to decompose the jobs and only assign deadlines to them. Figure~\ref{fig:workload} shows the variation in active workload over time with estimation of job lengths.


\begin{table*}[!t]
  \centering
  \caption{Cluster Sizes and Deadlines for Workload Classification for GCP}
   \begin{tabular}{c|r|r|r|r|r|r|r|r|r|r|c}
    \hline \hline
    Cluster & \multicolumn{5}{c|}{Workload A}        & \multicolumn{5}{c|}{Workload B}        & Deadline \\
    \hline
          & \# Jobs & \% Jobs & $S$ (MB) & $S'$ (MB) & $S''$ (MB) & \# Jobs & \% Jobs & $S$ (MB) & $S'$ (MB) & $S''$ (MB) & (\# slots) \\ \hline
    1     & 5691  & 96.56 & 0.02  & 0.00  & 0.67  & 6313  & 95.10 & 0.02  & 0.00  & 0.48  & 1 \\
    2     & 116   & 1.97  & 44856.77 & 15493.69 & 83.89 & 223   & 3.36  & 39356.46 & 6594.93 & 99.26 & 2 \\
    3     & 27    & 0.46  & 57121.85 & 148012.87 & 16090.40 & 41    & 0.62  & 110076.24 & 282.08 & 1.60  & 3 \\
    4     & 23    & 0.39  & 125953.59 & 0.00  & 51.89 & 25    & 0.38  & 379363.01 & 0.00  & 521.45 & 4 \\
    5     & 19    & 0.32  & 0.33  & 0.00  & 49045.29 & 16    & 0.24  & 0.04  & 0.00  & 40355.53 & 5 \\
    6     & 8     & 0.14  & 207984.10 & 414045.45 & 3095.56 & 7     & 0.11  & 132529.27 & 383548.19 & 31344.38 & 6 \\
    7     & 5     & 0.08  & 541522.77 & 0.00  & 0.05  & 4     & 0.06  & 258152.65 & 1020741.05 & 22631.52 & 7 \\
    8     & 3     & 0.05  & 0.05  & 0.00  & 203880.59 & 3     & 0.05  & 0.29  & 0.00  & 311410.40 & 8 \\
    9     & 1     & 0.02  & 7201446.27 & 48674.26 & 0.10  & 3     & 0.05  & 1182734.09 & 3.93  & 0.01  & 9 \\
    10    & 1     & 0.02  & 934594.27 & 8413335.44 & 0.06  & 3     & 0.05  & 0.56  & 0.00  & 622103.12 & 10 \\
    \hline
    \end{tabular}%
  \label{table:gcp}%
\end{table*}%

\subsubsection*{Deadline assignment}
For VFW($\delta$), the deadline $D$ is uniform and is assigned in terms of number of slots the workload can be delayed. For our simulation, We vary $D$ from $1-12$ slots which gives latency from 5 minutes upto 1 hour. This is realistic as deadlines of 8-30 minutes for MapReduce workload have been used in the literature \cite{n10,n11}. For GCP, we use k-means clustering to classify the workload into 10 groups based on the map, shuffle and reduce bytes ($S, S', S''$). The characteristics of each group are depicted in Table~\ref{table:gcp}. From Table~\ref{table:gcp}, it is evident that smaller jobs dominate the workload mix, as discussed in Section IIA. For each new class of jobs we assign a deadline from $1-10$ slots such that smaller class has larger deadline and larger class of jobs has smaller deadline. The deadline for a job should not be less than the length of a job. If the assigned deadline is less than the job length, we update the deadline to be equal to the length of the job.

\subsection{Analysis of the Simulation}
We now analyze the impact of different parameters on cost savings provided by VFW($\delta$) and GCP. We then compare VFW($\delta$) and GCP for uniform deadline (GCP-U).

\subsubsection*{Impact of deadline}
The first parameter we study is the impact of different deadline requirements of the workload on the cost savings. Figure~\ref{fig:deadline} shows that even for deadline $D$ as small as 2 slots, the cost is reduced by $\sim$40\% for GCP-U, $\sim$20\% for VFW($\delta$) while the offline algorithm gives a cost saving of $\sim$60\% compared to the naive algorithm. It also shows that for all the algorithms, large $D$ gives more cost savings as more workload can be delayed to reduce the variation in the workload. As $D$ grows larger the cost reduction from GCP-U and VFW($\delta$) approaches offline cost saving which is as much as 70\%. For VFW($\delta$), the cost saving is always less than GCP-U for both the workload.


\begin{figure}[!t]
\centerline{\subfigure[Workload A]{\includegraphics[width
=1.7in]{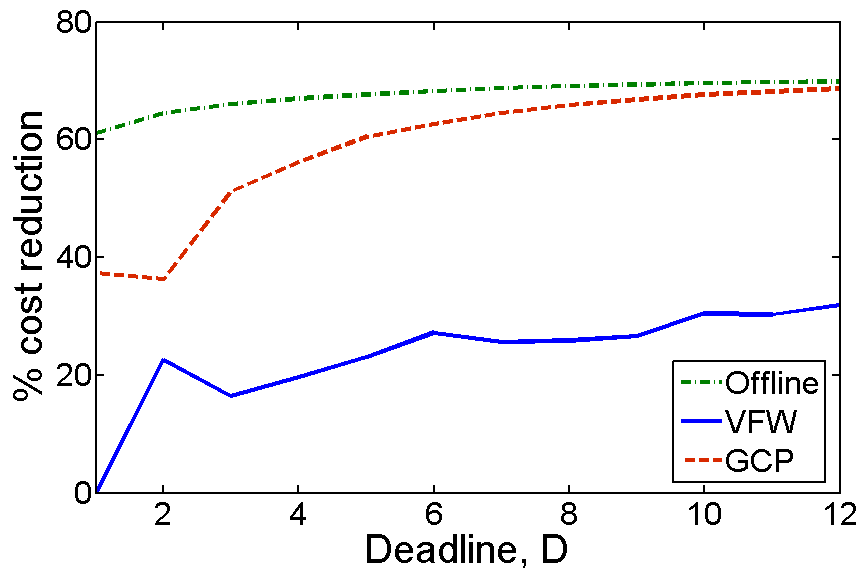}
}
\hfil
\subfigure[Workload B]{\includegraphics[width=1.7in]{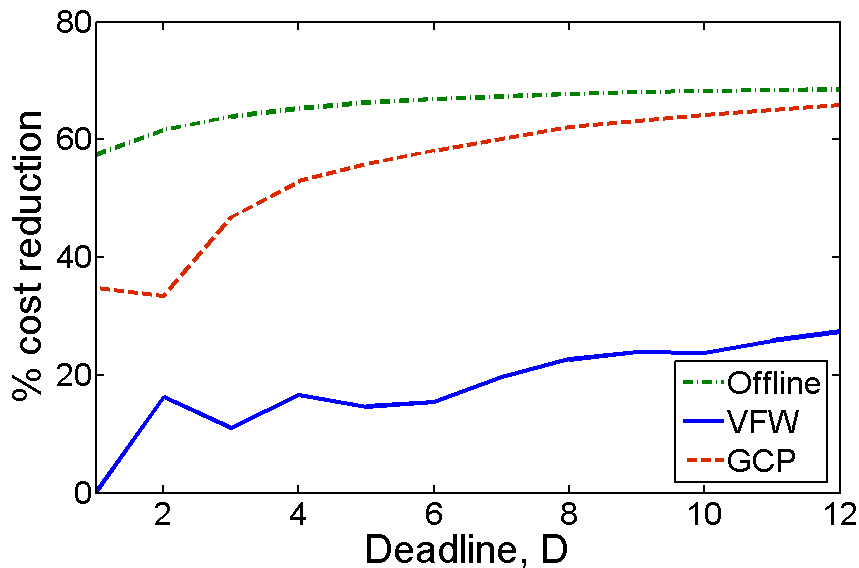}
}}
\caption{Impact of deadline on cost incurred by GCP-U, Offline and VFW($\delta$) with $\delta=D/2$. }
\label{fig:deadline}
\end{figure}

\subsubsection*{Impact of $\delta$ for VFW($\delta$)}
The parameter $\delta$ is used as a lookahead to detect a valley in the VFW($\delta$) algorithm. If $\delta$ is large, valley detection performs well but it may be too late to fill the valley due to the deadlines. On the other hand if $\delta$ is small, valley detection does not work well because the capacity has already gone down to the lowest value. Figure~\ref{fig:valleydetect} illustrates the valley detection for small $\delta$ and large $\delta$. Although the cost savings from VFW($\delta$) largely depends on the nature of the workload curve, Figure~\ref{fig:delta} shows that $\delta\sim D/2$ is a conservative estimate for better cost savings.

\begin{figure}[!t]
\begin{center}
\includegraphics[width=3.0in]{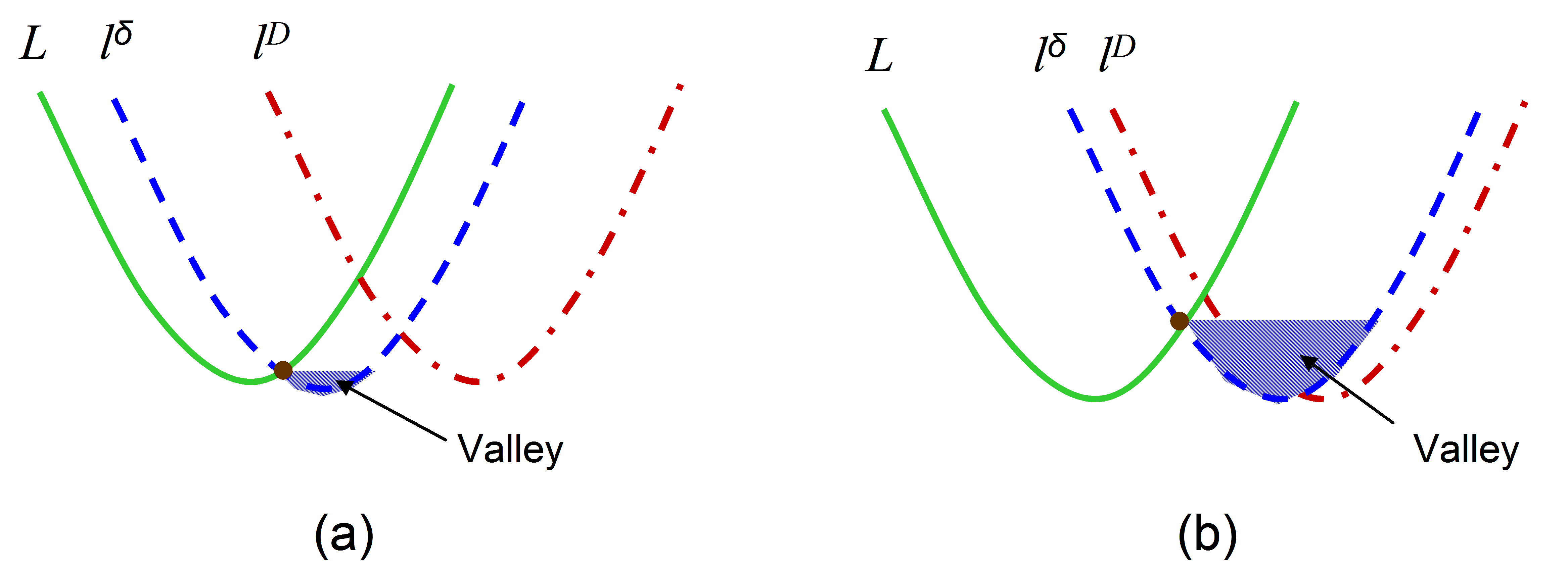}
\caption{Valley detection for (a) small $\delta$ and (b) large $\delta$ for VFW($\delta$).}
\label{fig:valleydetect}
\end{center}
\end{figure}

\begin{figure}[!t]
\centerline{\subfigure[Workload A]{\includegraphics[width
=1.7in]{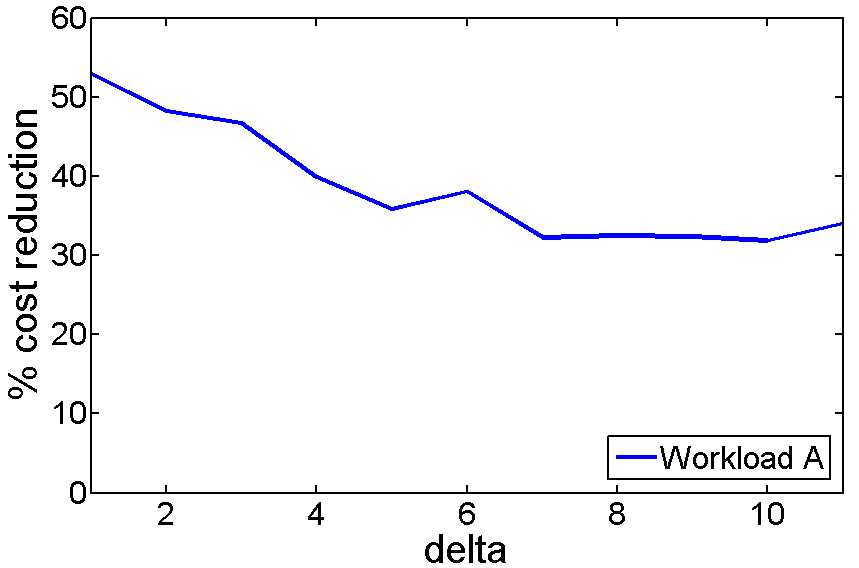}
}
\hfil
\subfigure[Workload B]{\includegraphics[width=1.7in]{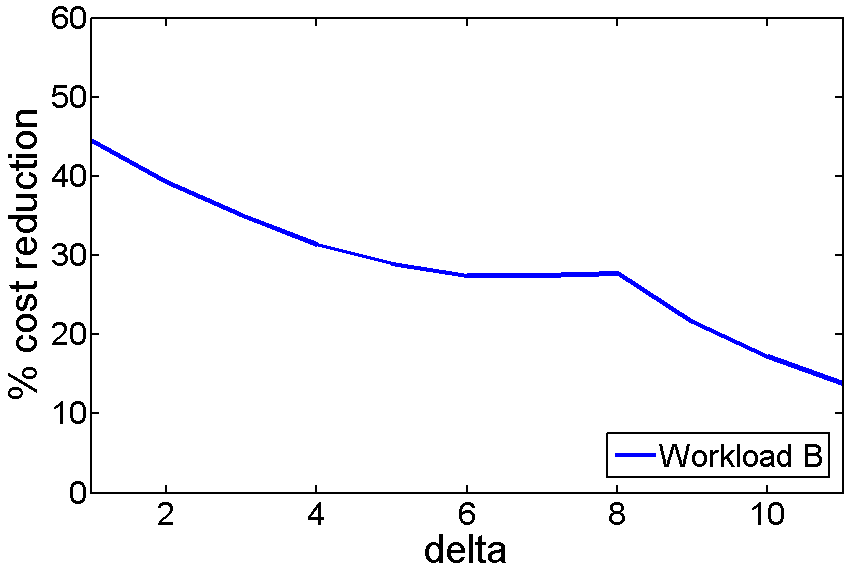}
}}
\caption{Impact of $\delta$ for VFW($\delta$) with deadline $D=12$. }
\label{fig:delta}
\end{figure}

\subsubsection*{Performance of GCP}
We evaluated the cost savings from GCP by assigning different deadline by classifying the workload as shown in Table~\ref{table:gcp}. For conservative estimates of deadline requirements (1-10), we found 47.66\% cost reduction for Workload A and 45.65\% cost reduction for Workload B each of which remains close to the offline optimal solutions.


\subsubsection*{Comparision of VFW($\delta$) and GCP}
We compare GCP for uniform deadline (GCP-U) with VFW($\delta$) for $\delta=D/2$. Figure~\ref{fig:deadline} illustrates the cost reduction for VFW($\delta$) and GCP-U with different deadlines $D = 1-12$. For both the workload, GCP-U performs better than VFW($\delta$). However for some workload, valley filling with workload as in VFW($\delta$) can be more beneficial than provisioning capacity for $D$ consecutive slots as in GCP. Hence we conclude that the comparative performance of the online algorithms depends largely on the nature of the workload. Since both the algorithms are based on linear program, they take around 10-12ms to compute schedule at each step.

\section{Experimentation}
In this section, we validate our algorithms on MapReduce workload by provisioning capacity on a Hadoop cluster. We evaluate the cost-savings by energy consumption calculated from common power model using different measured metrics.

\subsection{Experimental Setup}
We setup a Hadoop cluster (version 0.20.205) consisting of 35 nodes on Amazon's Elastic Compute Cloud (EC2) \cite{n1,n2}. Each node in the cluster is a small instance with 1 virtual core, 1.7 GB memory, 160 GB storage. We configured one node as master and four core nodes to contain the Hadoop DFS and the other 30 nodes as task nodes. The provisioning is done on the task nodes dynamically. We used the Amazon Elastic MapReduce service for provisioning capacity on the Hadoop cluster. Amazon Elastic MapReduce takes care of provisioning machines and migration of tasks between machines while keeping all data available. We used the Statistical Workload Injector for MapReduce (SWIM) \cite{n4} to generate the MapReduce workload for our cluster using the Facebook traces from Figure~\ref{fig:workload}(a). We run our experiment for 4 hours with slot length of 5 minutes. For the traces of Figure~\ref{fig:workload}(a), 602 jobs were released in the first 48 slots.

We first schedule the jobs and provision the task nodes by the `follow the workload' strategy. We then schedule the same jobs and provision the task nodes using our algorithms as illustrated in Figure~\ref{fig:solution}. In order to make comparison between VFW($\delta$) and GCP algorithms we used a uniform deadline of 10 minutes ($D=2$). In each of the experiments, we measured the seven metrics (available from Amazon Cloudwatch) for each of the `running' nodes in each time slot over the time interval of 4 hours and 10 minutes (50 slots). In the last 2 slots, the capacity of the task nodes were provisioned to zero for the `follow the workload' algorithm while our algorithms execute the delayed workload in those slots. All the jobs released in the first 48 slots were completed before the end of 50th slot. The seven metrics that are available for measurement for each virtual machines are: CPUUtilization, DiskReadBytes, DiskReadOps, DiskWriteBytes, DiskWriteOps, NetworkIn and NetworkOut.

%
%
%
%
%
%
%
%
%

\begin{figure*}[!t]
\centerline{\subfigure[Follow the workload]{\includegraphics[width =
2.2in]{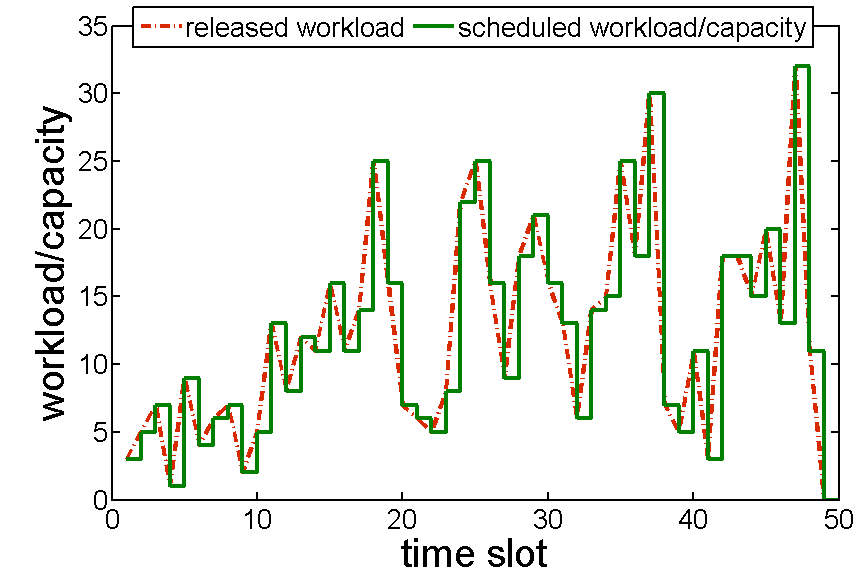}
}
\hfil
\subfigure[VFW($\delta$)]{\includegraphics[width=2.2in]{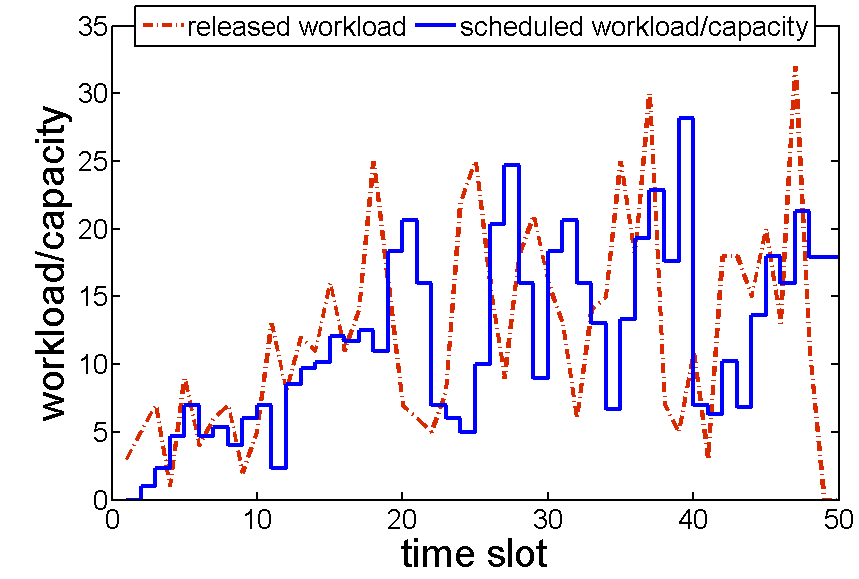}
}
\hfil
\subfigure[GCP-U]{\includegraphics[width =2.2in]{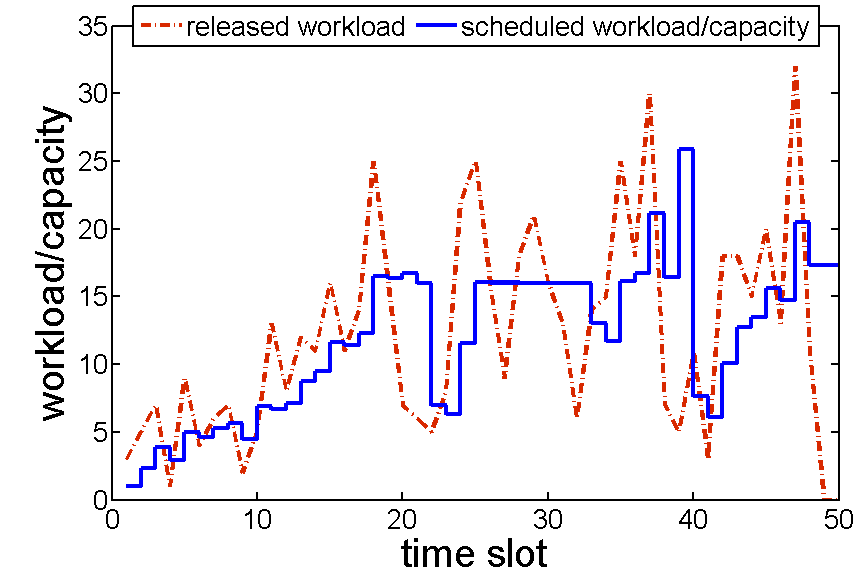}
}
}
\caption{The solutions for (a) Follow the workload, (b) VFW($\delta$) and (c) GCP-U algorithms with uniform deadline $D=2$ slots, $\delta=1$ and time slot = 5 mins.}
\label{fig:solution}
\end{figure*}

\subsection{Experimental Results}
We now discuss the results from the experimentation and compare the energy consumption between different algorithms.

\subsubsection*{Power Metering}
We use the general power model to evaluate energy consumption for the algorithms \cite{n7,n8}. The energy consumed by a virtual machine is represented as the sum of the energy consumption for utilization, disk operations and network IO,
\begin{equation}
\label{eqn:energy1}
E_{vm}(T) = E_{util,vm}(T) + E_{disk,vm}(T) + E_{net,vm}(T)
\end{equation}
where the energy consumption is over the duration $T$. The energy consumption for each of the components over a time slot $t$ (of length $\tau$) can be computed by the following equations.
\begin{IEEEeqnarray}{lll}
E_{util,vm}(t) &= & \alpha_{cpu} u_{cpu}(t) + \gamma_{cpu} \\
E_{disk,vm}(t) &= & \alpha_{rb} b_r(t) + \alpha_{wb} b_w(t) \nonumber\\
& & + \alpha_{ro} n_r(t) + \alpha_{wo} n_w(t) + \gamma_{disk} \nonumber\\
E_{net,vm}(t) &= & \alpha_{ni} b_{in}(t) + \alpha_{no} b_{out}(t) + \gamma_{net}\nonumber
\end{IEEEeqnarray}
where $u_{cpu}(t)$ is the average utilization, $b_r(t)$ and $b_w(t)$ are the total bytes read and written to disk, $n_r(t)$ and $n_w(t)$ are the total number of disk read and writes and $b_{in}(t)$ and $b_{out}(t)$ are the total bytes of network IO for the virtual machine over the time interval $t$. Since the difference in energies for disk read and write and network input and output are negligible \cite{n7}, we use common parameters $b_{db}(t)$, $b_{do}(t)$, and $b_{net}(t)$ by taking the sum of the respective values.  We normalize each of these values with their respective maximum values (in the interval $T$) so that each of these become a fraction between 0 and 1 and can be put in equation~(\ref{eqn:energy1}),
\begin{IEEEeqnarray}{lll}
\label{eqn:energy2}
E_{vm}(t) &= & \alpha_{cpu} u_{cpu}(t) + \gamma_{cpu} \\
& & + \alpha_{disk} u_{disk}(t) + \gamma_{disk} + \alpha_{dops} u_{dops}(t) + \gamma_{dops} \nonumber\\
& & + \alpha_{net} u_{net}(t) + \gamma_{net} \nonumber
\end{IEEEeqnarray}
where $u_{disk}(t)$, $u_{dops}(t)$ and $u_{net}(t)$ represent the normalized values of $b_{db}(t)$, $b_{do}(t)$, and $b_{net}(t)$ respectively.
If $m_t$ machines are active at time slot $t$, then the total energy consumed over the time interval $T$ can be computed using the following equation:
\begin{equation}
E(T) = \sum_{t=1}^T \sum_{i=1}^{m_t} E_{i}(t) * \frac{\tau}{3600} \text{  Watt-hour}
\end{equation}
where $E_{i}(t)$ is the energy consumed at machine $i$ over time slot $t$. To compute energy consumption, we used parameters from \cite{n8} listed in Table~\ref{table:parameter}. Typical values are used for cpu utilization, disk I/O and network I/O. Idle disk/network powers are negligible with respect to dynamic power and scale of workload.


\begin{table}[!ht]
  \centering
  \caption{Power Model Parameters}
    \begin{tabular}{c|l|c}
    \hline\hline
    Parameter & Comment & Value   \\
    \hline
    $\alpha_{cpu}$ &  Scaling factor: Utilization    &   25.70      \\
    $\alpha_{disk}$ &   Scaling factor: Disk Rd/Wr.     &   7.21        \\
    $\alpha_{dops}$ &  Scaling factor: Disk Op.     &      0     \\
    $\alpha_{net}$ &  Scaling factor: Network IO     &      0.66      \\
    $\gamma_{cpu}$ &   Idle cpu power consump.    &    60.30         \\
    $\gamma_{disk}$, $\gamma_{dops}$ &  Idle disk power consump.    &    0      \\
    $\gamma_{net}$ &   Idle network power consump.   &        0    \\
    \hline
    \end{tabular}%
  \label{table:parameter}%
\end{table}%

The total energy consumption and the \% reduction with respect to `follow the workload' in each of the metrics for different schedules are illustrated in Table~\ref{table:reduction}. For the period of 4 hours 10 minutes (50 slots), GCP algorithm gives energy reduction of $\sim$12\% which is significantly better than the reduction of $\sim$6.02\% from the VFW($\delta$) algorithm. The reductions from both the algorithms are far better with respect to workload schedule without provisioning. Table~\ref{table:reduction} also shows that variation in CPU utilization, Disk I/O and Network I/O across different algorithms. This variation results from the difference in capacity provisioning across algorithms that changes migration of jobs and disk I/O in the cluster. Figure~\ref{fig:energyconsumption} illustrates the average energy consumption within each slot over the time interval showing significant reduction in the peak energy consumption. As the provisioning algorithms cut off peaks from the workload and provision the machines without wasting computation capacity, they reduce the peak energy consumption for the data center.

\begin{table*}[!t]
  \centering
  \caption{Total energy consumption and the total values for different Metrics from the cluster for different schedule}
    \begin{tabular}{l|c|c|c|c|c|c}
    \hline\hline
    Metrics & No Provisioning & Follow  & VFW($\delta$) &  \% Reduction & GCP-U &  \% Reduction \\
    \hline
    Energy Consumption(kWh) &   8.60    &   4.46    &  4.19     &    6.02   & 3.93 & 11.96 \\
    CPUUtilization(sum) &   32505.95    &   22805.98    &  21014.51     &     7.86  & 20400.02& 10.55 \\
    DiskReadBytes(GB) &   0.25    &   12.95    &   7.56    &     41.64  & 3.85 & 70.29 \\
    DiskWriteBytes(GB) &  10.42     &   8.01    &   8.44    &    -5.48   & 6.55 & 18.19\\
    DiskReadOps(count) &  18883     &   1109320    &   710451    &    35.96 & 396070 & 64.30   \\
    DiskWriteOps(count) &  1746347     &   1134108    &   1020343    &    10.03 & 901860 &  20.48  \\
    NetworkIn(GB) &  45.69     &   42.30    &   43.69    &    -3.29  & 42.88 & -1.38 \\
    NetworkOut(GB) &  44.21     &  42.45     &   38.64    &    8.97  & 41.48 &  2.29\\
    \hline
    \end{tabular}%
  \label{table:reduction}%
\end{table*}%





\begin{figure}[!t]
\begin{center}
\includegraphics[width=2.5in]{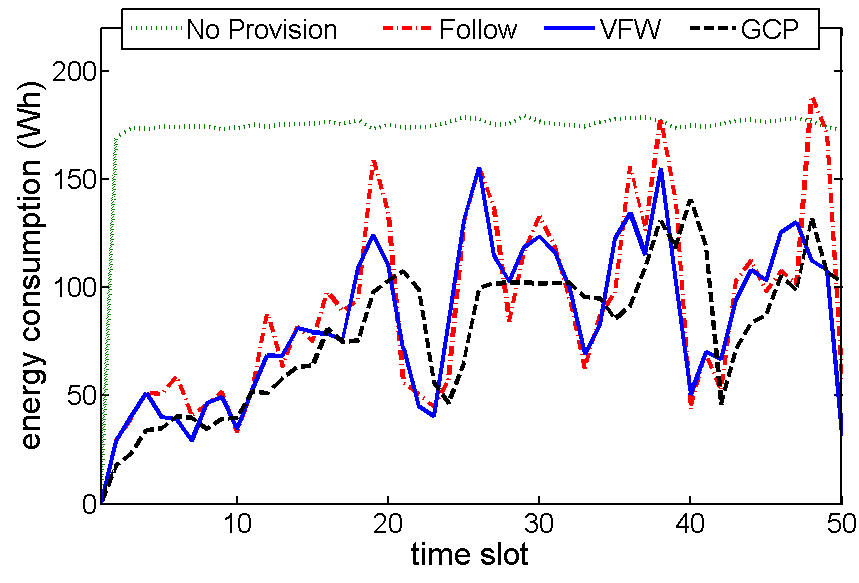}
\caption{Average energy consumption from the cluster with time slots of 5 minutes, over a period of 4 hours.}
\label{fig:energyconsumption}
\end{center}
\end{figure}

\section{Related Work}
With the importance of energy management in data centers, many scholars have applied energy-aware scheduling because of its low cost and practical applicability. In energy-aware scheduling, most work tries to find a balance between energy cost and performance loss through DVFS (Dynamic Voltage and Frequency Scaling) and DPM (Dynamic Power Management), which are the most common system-level power saving methods.  Beloglazov et al. \cite{1} give the taxonomy and survey on energy management in data centers. Dynamic capacity provisioning is part of DPM technique.  Chase et al. \cite{2} introduce the executable utility functions to quantify the value of performance and use economic approach to achieve resource provisioning.  Pinheiro et al. \cite{3} consider resource provisioning in both application and operating system level. They dynamically turn on or turn off nodes to adapt to the changing load, but do not consider the switching cost.

Most work on dynamic capacity provisioning for independent workload uses models based on queueing theory \cite{4,8}, or control theory \cite{7,29}. Recently Lin et al. \cite{9} used more general and common energy model and delay model and designed a provisioning algorithm for service jobs (e.g. HTTP requests) considering switching cost for the machines. They proposed a lazy capacity provisioning (LCP) algorithm which dynamically turns on/off servers in a data center to minimize energy cost and delay cost for scheduling workload. However their algorithm does not perform well for high peak-to-mean ratio (PMR) of the workload and does not provide bound on maximum delay. Moreover, LCP aims at minimizing the average delay while we regard latency as the deadline constraint. Instead of penalizing the delay, we purposely defer jobs within deadline in order to reduce the switching cost of the servers.


Many applications in real world require delay bound or deadline constraint e.g. see Lee et al. \cite{24}. In the context of energy conservation, deadline is usually a critical adjusting tool between performance loss and energy consumption. Energy efficient deadline scheduling was first studied by Yao et al. \cite{10}. They proposed algorithms, which aim to minimize energy consumption for independent jobs with deadline constraints on a single variable-speed processor. After that, a series of work was done to consider online deadline scheduling in different scenarios, such as discrete-voltage processor, tree-structured tasks, processor with sleep state and overloaded system \cite{11,12}. In the context of data center, most prior work on energy management merely talks about minimizing the average delay without any bound on the delay. Recently, Mukherjee et al. \cite{17} proposed online algorithms considering deadline constraints to minimize the computation, cooling and migration energy for machines.  Goiri et al. \cite{Goiri} considered only batch jobs and proposed GreenSlot, a scheduler with deadline requirements for the jobs. GreenSlot predicts the amount of solar energy that will be available in near future and schedules the workload to maximize the green energy consumption while meeting the jobs' deadline. However these works are on job assignment problem and not on dynamic resource provisioning problem, where the number of needed servers is given in advance.

Recently researchers have used scheduling with deferral to improve performance of MapReduce jobs \cite{n6,n11}. Although MapReduce was designed for batch jobs, it has been increasingly used for small time-sensitive jobs. Delay scheduling with performance goals was proposed by Zaharia et al. \cite{n6} for scheduling jobs inside a Hadoop cluster with given resources. Verma et al. introduced a SLA-driven scheduling and resource provisioning framework considering given soft-deadline requirements for the MapReduce jobs \cite{n9,n10}. In contrast to these works, we consider hard-deadlines and schedule jobs within those deadlines and provision capacity to save energy.
Recently, Chen et al. \cite{n5} identified a large class of interactive MapReduce workload and proposed policies for scheduling batch and small interactive jobs in separate cluster without any provisioning mechanism for the machines in the cluster. In contrast, we propose provisioning algorithms for the mix of batch and interactive jobs under bounded latency with constant competitive ratio.


%
%




\section{Conclusion}
We have shown that significant reduction in energy consumption can be achieved by dynamic deferral of workload for capacity provisioning inside data centers. We have proposed two new algorithms, VFW($\delta$) and GCP, for provisioning the capacity and scheduling the workload while guaranteeing the deadlines. The algorithms take advantage from the flexibility in the latency requirements of the workloads for energy savings and guarantee bounded cost and bounded latency under very general settings - arbitrary workload, general deadline and general energy cost models. Further both the VFW($\delta$) and GCP algorithms are simple to implement and do not require significant computational overhead. Additionally, the algorithms have constant competitive ratios and offer noteworthy cost savings as proved by theory and demonstrated by simulations respectively. We have validated our algorithms on MapReduce workload by provisioning capacity on a Hadoop cluster. For a small interval of 4 hours, we found $\sim$6.02\% total energy savings for VFW($\delta$) and $\sim$12\% for GCP with respect to the naive `follow the workload' approach. Both the algorithms achieve more than 50\% reduction in energy consumption with respect to `no provisioning' which is a common practice for the current data center providers. Although we have used MapReduce workload for validation, our algorithms can be applied for any workload as data centers have separate (physical/virtual) clusters for MapReduce and non-MapReduce jobs. The provisioning can be done on each such cluster. In order to reduce the energy consumption, the data center providers should provision their capacity (physical/virtual) and utilize the flexibilities from SLAs via dynamic deferral.


\section*{Acknowledgment}
This work was sponsored in part by the Multiscale Systems
Center (MuSyC) and NSF Variability Expedition.

{\small{

}}

\end{document}